\documentclass[conference]{IEEEtran}
\IEEEoverridecommandlockouts
\usepackage{cite}
\usepackage{amsfonts}
\usepackage{algorithmic}
\usepackage{textcomp}
\usepackage{xcolor}
\usepackage[utf8]{inputenc}
\usepackage[T1]{fontenc}
\usepackage{microtype}

\usepackage{verbatim} 
\usepackage{graphicx} 
\usepackage{url,hyperref}
\usepackage{listings}
\usepackage{array}

\usepackage{amsthm}
\usepackage{amsmath}
\usepackage{amssymb} 
\usepackage{multicol} 
\usepackage{lastpage}
\usepackage{caption}
\usepackage{comment}
\usepackage{xspace}
\usepackage{breakcites}
\usepackage{wrapfig}
\usepackage{relsize} 
\usepackage{bm}  

\usepackage{textgreek}
\usepackage{mathtools} 
\usepackage[inline]{enumitem}

\usepackage{flushend}

\newcommand{\textapprox}{\raisebox{0.5ex}{\texttildelow}}

\makeatletter
\newcommand{\linebreakand}{%
  \end{@IEEEauthorhalign}
  \hfill\mbox{}\par
  \mbox{}\hfill\begin{@IEEEauthorhalign}
}
\makeatother

\lstdefinestyle{interfaces}{
  float=tp,
  floatplacement=tbp
}

\usepackage{url}

\newcommand{\eg}{\emph{e.g.,}\xspace}

\newcommand{\ie}{{\em i.e.,}\xspace}
\newcommand{\etal}{{\em et al.}\xspace}

\newcommand{\etc}{{\em etc.}\xspace}
\newcommand{\Fig}[1]{Fig.~\ref{fig:#1}} 

\newcommand{\Sec}[1]{\S\ref{sec:#1}\xspace}
\newcommand{\Eqn}[1]{Eqn. \ref{#1}\xspace}

\newcommand{\nop}[1]{}
\newcommand{\Lst}[1]{Listing ~\ref{lst:#1}}

\newcommand{\ithbit}[3]{#1_#2[#3]}
\newcommand{\ithbitx}[2]{#1[#2]}

\newcommand{\Obs}[1]{Observation~\ref{obs:#1}}
\newcommand{\Lemma}[1]{Lemma~\ref{#1}\xspace}
\newcommand{\OurTheorem}[1]{Theorem~\ref{#1}\xspace}
\newenvironment{parafont}{\fontfamily{ptm}\selectfont}{}
\theoremstyle{remark}

\newtheorem{theorem}{\textbf{Theorem}}
\newtheorem{lemma}[theorem]{\textbf{Lemma}}
\newtheorem{definition}[theorem]{\textbf{Definition}}
\newtheorem{observation}[theorem]{\textbf{Observation}}
\newcommand{\Para}[1]{\vspace{4pt}\noindent\begin{parafont}\fontsize{10.0}{12.00}\textbf{\textit{#1}}\end{parafont}}

\newcommand\delete[1]{}

\newcommand{\code}[1]{\texttt{\detokenize{#1}}}

\newcommand{\bitwiseAnd}{{\texttt{\&}}}
\newcommand{\bitwiseOr}{{\texttt{|}}}
\newcommand{\bitwiseXor}{\raisebox{.4ex}{\texttt{$\mathsmaller{\mathsmaller{\bigoplus}}$}}}
\newcommand{\bitwiseNot}{{\texttt{\textapprox}}}
\newcommand{\bitwiseRshift}{{\texttt{>{}>}}}
\newcommand{\bitwiseLshift}{\,\texttt{<{}<}\,}
\newcommand{\arithAdd}{\,\texttt{+}\,}

\newcommand{\arithMul}{\,\texttt{*}\,}

\newcommand{\opC}{\texttt{op}_{\mathbb{C}}}
\newcommand{\opT}{\texttt{op}_{\mathbb{T}}}

\newcommand{\proofheading}[1]{\noindent\textbf{#1}}
\newcommand{\proofsubheading}[1]{\noindent\textit{#1}}

\newcommand{\Tnumin}[2]{#1 \in \gamma(#2)}
\newcommand{\Tnumrshift}[2]{\code{tnum_rshift}(#1, #2)}
\newcommand{\Tnumlshift}[2]{\code{tnum_lshift}(#1, #2)}

\newcommand{\nbit}{$n$-bit\xspace}
\newcommand{\nth}[1]{#1^{th}\xspace}
\newcommand{\ibit}[1]{$#1$-bit\xspace}
\newcommand{\ntrit}{$n$-trit\xspace}

\newcommand{\ith}[1]{$#1^{th}$}

\newcommand{\GaliosConnectionGeneric}{
$(\mathbb{C}, \leqconc) \,
\substack{\xrightarrow[\text{$\;\alpha\;$}]{} \\[-1em] \\[-1em]
\xleftarrow[]{\text{$\;\gamma\;$}} } \,
(\mathbb{A}, \leqabst)$\xspace}

\newcommand{\GaliosConnectionIntegerToBitField}{$
(2^{\mathbb{Z}_n}
\substack{\xrightarrow[\text{$\;\alpha\;$}]{} \\[-1em] \\[-1em]
\xleftarrow[]{\text{$\;\gamma\;$}} } \; \mathbb{Z}_n \times
\mathbb{Z}_n)
$}

\newcommand{\tnumi}[1]{({#1}.{v}, {#1}.{m})}

\newcommand{\tvi}[1]{{#1}.{v}}
\newcommand{\tmi}[1]{{#1}.{m}}
\newcommand{\tnumival}[3]{({#1}.{v}={#2}, {#1}.{m}={#3})}
\newcommand{\wellformedPredicate}{wellformed}
\newcommand{\memberPredicate}{member}
\newcommand{\addPredicate}{add}

\newcommand{\bb}[2]{\mathbb{#1}_{#2}}

\newcommand{\ourmul}{\code{our_mul}\xspace}
\newcommand{\Ourmul}{\code{Our_mul}\xspace}
\newcommand{\ourmulsimplified}{\code{our_mul_simplified}\xspace}

\newcommand{\reghermul}{\code{bitwise_mul}\xspace}

\newcommand{\kernmul}{\code{kern_mul}\xspace}
\newcommand{\opone}{\texttt{op\textsubscript{1}}\xspace}
\newcommand{\optwo}{\texttt{op\textsubscript{2}}\xspace}

\newcommand{\ourmulfasterthankernmulby}{33\%\xspace}
\newcommand{\ourmulfasterthanreghermulby}{32\%\xspace}

\newcommand{\leqconc}{\sqsubseteq_{\mathbb{C}}}
\newcommand{\leqabst}{\sqsubseteq_{\mathbb{A}}}

\newcommand{\Zpowerset}{2^{\mathbb{Z}_n}}

\def\compactify{\itemsep=0pt \topsep=0pt \partopsep=0pt \parsep=0pt \leftmargin=12pt}
\let\latexusecounter=\usecounter
\newenvironment{CompactItemize}
  {\def\usecounter{\compactify\latexusecounter}
   \begin{itemize}}
  {\end{itemize}\let\usecounter=\latexusecounter}
\newenvironment{CompactEnumerate}
  {\def\usecounter{\compactify\latexusecounter}
   \begin{enumerate}}
  {\end{enumerate}\let\usecounter=\latexusecounter}

\newcommand{\percentageofourmulkernmulequalinputsatbitwidtheight}{$99.92$}
\newcommand{\Supplement}{extended technical report}

\begin{document}

\title{
  Sound, Precise, and Fast Abstract Interpretation with Tristate Numbers
}

\author{\IEEEauthorblockN{Harishankar Vishwanathan\IEEEauthorrefmark{1},
Matan Shachnai\IEEEauthorrefmark{2}, Srinivas Narayana\IEEEauthorrefmark{3} and
Santosh Nagarakatte\IEEEauthorrefmark{4}}
\IEEEauthorblockA{\textit{Rutgers University, USA}\\
Email: \IEEEauthorrefmark{1}harishankar.vishwanathan@rutgers.edu,
\IEEEauthorrefmark{2}mys35@cs.rutgers.edu,
\IEEEauthorrefmark{3}srinivas.narayana@rutgers.edu, \\
\IEEEauthorrefmark{4}santosh.nagarakatte@cs.rutgers.edu}}
\maketitle

\begin{abstract}
Extended Berkeley Packet Filter (BPF) is a language and
run-time system that allows non-superusers to extend the
Linux and Windows operating systems by downloading user code
into the kernel. To ensure that user code is safe to run in
kernel context, BPF relies on a static analyzer that proves
properties about the code, such as bounded memory access and
the absence of operations that crash. The BPF static
analyzer checks safety using abstract interpretation with
several abstract domains. Among these, the domain of tnums
(tristate numbers) is a key domain used to reason about
the bitwise uncertainty in program values.
This paper formally
specifies the tnum abstract domain and its arithmetic
operators. We provide the first proofs of soundness and
optimality of the abstract arithmetic operators for tnum
addition and subtraction used in the BPF analyzer.  Further,
we describe a novel sound algorithm for multiplication of
tnums that is more precise and efficient (runs
\ourmulfasterthankernmulby faster on average) than the Linux
kernel's algorithm. Our tnum multiplication is now merged in
the Linux kernel.

\begin{IEEEkeywords}
    Abstract domains, Program verification, Static analysis,
    Kernel extensions, eBPF
  \end{IEEEkeywords}    
\end{abstract}

\section{Introduction}

Static analysis is an integral part of
compilers~\cite{clang-analysis-toolset,
  gcc-and-static-analysis, testing-static-analysis-cgo20,
  gcc-static-analysis-options}, sandboxing
technologies~\cite{vera-pldi20, rocksalt-pldi12,
  untrusted-extensions-pldi19}, and continuous integration
testing~\cite{coverity-cacm10}. For example, static analysis
may be used to prove that the value of a program variable
will always be bounded by a known constant, allowing a
compiler to eliminate dead code~\cite{alive-infer-pldi17} or
a sandbox to remove an expensive run-time
check~\cite{vera-pldi20}.

Our work is motivated by static analysis in the context of
{\em Berkeley Packet Filter (BPF)}, a language and run-time
system~\cite{lwn-intro-to-ebpf, bpf-kernel-documentation}
that enables users to extend the functionality of the Linux
and Windows operating systems without writing kernel
code. BPF is widely deployed in production systems
today~\cite{netflix-netconf-day-1, pixie-labs,
  cloudflare-l4drop, katran-facebook-talk, cilium-lb,
  suricata, cilium}. BPF uses a static analyzer to validate
that user programs are {\em safe} before they are executed
in kernel context~\cite{bpf-kernel-documentation,
  untrusted-extensions-pldi19}: the analyzer must be able to
show that the program does not access unpermitted memory
regions, does not leak privileged kernel data, and does not
crash. If the analyzer is unable to prove these properties,
the user program is {\em rejected} and cannot execute in
kernel context.

BPF static analysis must be {\em sound, precise,} and {\em
  fast}.

\begin{CompactItemize}
\item {\em Soundness:} Unsound analysis that accepts
  malicious code may result in arbitrary read-write
  capabilities for users in the
  kernel~\cite{seriousness-of-range-tracking}. Unfortunately,
  the Linux static analyzer has been a source of numerous
  such bugs in the past~\cite{kernel-pwning-bpf,
    manfred-paul-privilege-escalation, get-rekt-hardened,
    alu64-shifts-bug, bpf-bounds-calculation-bug,
    bpf-bug-32bit-alu-ops-shifts,
    bpf-bug-incorrect-arsh-simulation,
    bpf-bug-precision-tracking-error, zero-shift-bug,
    bpf-integer-overflows, bpf-bug-check-stack-boundary,
    bpf-bug-force-stack-alignment-checks,
    state-type-comparison-pointer-leak-fix,
    bpf-bug-incorrect-argument-check}.
  \item {\em Precision: } To provide a usable system, the
    analyzer must not reject safe programs due to
    imprecision in its analysis. Users often need to rewrite
    their programs to get their code past the
    analyzer~\cite{untrusted-extensions-pldi19,
      doc-on-bpf-verifier-complexity-1,
      doc-on-bpf-verifier-complexity-2}.
  \item {\em Speed:} The analyzer must keep the time and
    overheads to load a BPF program
    minimal~\cite{bpf-kernel-documentation, bpf-design-qa,
      bpf-increase-complexity-limit}. Programs are often
    used to trace systems running heavy workloads.
\end{CompactItemize}

The BPF static analyzer employs abstract
interpretation~\cite{cousot-1977} with multiple abstract
domains to track the types, liveness, and values of program
variables across all executions. One of the key abstract
domains, termed {\em tristate numbers} or {\em tnums} in the
Linux kernel~\cite{tnum-kernel-source}, tracks which bits of
a value are known to be 0, known to be 1, or unknown
(denoted $\mu$) across executions. For example, a 4-bit
variable $x$ abstracted to $01\mu0$ can take on the binary
values $0100$ and $0110$. The analyzer can infer that the
expression $x \leq 8$ will always return $true$, and use
this fact later to show the safety of a memory access.

The kernel provides algorithms to implement bit-wise
operations such as and ($\bitwiseAnd$), or ($\bitwiseOr$),
and shifts ($\bitwiseLshift, \bitwiseRshift$) over
tnums. The kernel also provides efficient algorithms for
arithmetic (addition, subtraction, and multiplication) over
tnums. In particular, addition and subtraction run in $O(1)$
time over \nbit program variables given \nbit machine
arithmetic instructions.

Unfortunately, the kernel provides no formal reasoning or
proofs of soundness or precision of its algorithms. Prior
works that explored abstract domains for bit-level
reasoning~\cite{monniaux-device-driver-verification,
  mine-bitfield, regehr-deriving-abstract-transfer,
  dataflow-pruning-oopsla20, testing-static-analysis-cgo20,
  llvm-known-bits-analysis} provide sound and precise
abstract operators for bit-wise operations ($\bitwiseAnd,
\bitwiseOr, \bitwiseRshift$, \etc). The only arithmetic
algorithms we are aware
of~\cite{regehr-deriving-abstract-transfer} are much slower
than the kernel's algorithms (\Sec{background}). Arithmetic
operations are tricky to reason about as they propagate
uncertainty across bits in non-obvious ways.  For example,
suppose $a$ is known to be the \nbit constant $11\cdots1$
and $b$ is either $0$ or $1$ across all executions. Only one
bit is uncertain among the operands, yet all bits in $a+b$
are unknown, since $a + b$ can be either $11\cdots1$ or
$00\cdots0$.

This paper makes the following contributions
(\Sec{verification}). We provide the first proofs of {\em
soundness} and {\em optimality} (\ie maximal
precision~\cite{testing-static-analysis-cgo20,
mine-tutorial}) of the kernel's algorithms for addition and
subtraction. We believe this result is remarkable for
abstract operators exhibiting $O(1)$ run time and reasoning
about uncertainty across bits. We were unable to prove the
soundness of the kernel's tnum multiplication. Instead, we
present a novel multiplication algorithm that is provably
sound. It is also more precise and
\ourmulfasterthanreghermulby faster than prior
implementations~\cite{regehr-deriving-abstract-transfer,
tnum-kernel-source}. This algorithm is now merged into the
latest Linux kernels. Our reproducible artifact is publicly
available~\cite{vishwanathan-tnums-doi}.

\section{Background}
\label{sec:background}

The BPF static analyzer in the kernel checks the safety of BPF
programs by performing abstract interpretation using the tnum abstract
domain (among others).  In this section, we provide a primer on
abstract interpretation and describe the tnum abstract domain and its
operators.

\subsection{Primer on Abstract Interpretation}
Abstract interpretation \cite{cousot-1977} is a form of static
analysis that captures the values of program variables in all
executions of the program. Abstract interpretation employs {\em
  abstract values} and {\em abstract operators}. Abstract values are
drawn from an {\em abstract domain}, each element of which is a
concise representation of a set of concrete values that a variable may
take across executions.  For example, an abstract value from the
interval abstract domain~\cite{interval-domain} $\{[a, b] \mid a, b
\in \mathbb{Z}, a \leq b\}$ models the set of all concrete integer
values (\ie $x \in \mathbb{Z}$) such that $a \leq x \leq b$.

\Para{Abstraction and Concretization functions.} An {\em
  abstraction function $\alpha$} takes a concrete set and
produces an abstract value, while a {\em concretization
  function $\gamma$} produces a concrete set from an
abstract value. For example, the abstraction of the set
$\{ 2, 4, 5\}$ in the interval domain is $[2, 5]$, which
produces the set $\{2, 3, 4, 5\}$ when concretized.

Formally, the domains of the abstraction and concretization
functions are two partially-ordered sets (posets) that
induce a lattice structure. We denote the concrete poset
$\mathbb{C}$ with the ordering relationship among elements
$\leqconc$. Similarly, we denote the abstract poset
$\mathbb{A}$ with the ordering relationship $\leqabst$. For
example, the interval domain employs the concrete poset
$\mathbb{C} \triangleq 2^\mathbb{Z}$, the power set of
$\bb{Z}{}$, with the subset
relation $\subseteq$ (\eg $\{1, 2\} \subseteq \{1, 2, 3\}$)
as its ordering relation. The abstract poset is $\mathbb{A}
\triangleq \mathbb{Z} \times \mathbb{Z}$ with the ordering
relation $[a,b] \leqabst [c,d] \Leftrightarrow (c \leq a)
\wedge (d \geq b$).

A value $a \in \bb{A}{}$ is a {\em sound} abstraction of a
value $c \in \bb{C}{}$ if and only if $c \leqconc
\gamma(a)$. Moreover, $a$ is an {\em exact} abstraction of
$c$ if $c = \gamma(a)$.  Abstractions are often not exact,
over-approximating the concrete set to permit concise
representation and efficient analysis in the abstract
domain. For example, the interval $[2, 5]$ is a sound but
inexact abstraction of the set $\{2, 4, 5\}$.

\Para{Abstract operators} are functions over abstract values
which return abstract values. An abstract operator
implements an ``abstract version'' of a concrete operation
over concrete sets, hence enabling a static analysis to
construct the abstract results of program execution. For
example, abstract integer addition in the interval domain
(denoted $\mathbb{+}_{\mathbb{A}}$) abstracts concrete
integer addition (denoted $\mathbb{+}_{\mathbb{C}}$) as
follows: $ [a_1, b_1] \mathbb{+}_{\mathbb{A}} [a_2, b_2]
\triangleq [a_1 +_{\mathbb{C}} a_2, b_1 +_{\mathbb{C}}
  b_2]$. Abstract operators typically over-approximate the
resulting concrete set to enable decidable and fast analysis
at the expense of precision. For a concrete set $S \in
\bb{C}{}$, suppose we use the shorthand $f(S)$ to denote the
set $\{f(x)\ |\ x \in S\}$. An abstract operator $g:
\mathbb{A} \rightarrow \mathbb{A}$ is a {\em sound}
abstraction of a concrete operator $f: \mathbb{C}
\rightarrow \mathbb{C}$ if $\forall a \in \mathbb{A}:
f(\gamma(a)) \leqconc \gamma(g(a))$. Further, $g$ is {\em
  exact} if $\forall a \in A: f(\gamma(a)) = \gamma(g(a))$.

\Para{Galois connection.} Pairs of abstraction and
concretization functions $(\alpha, \gamma)$ are said to form
a Galois connection if~\cite{mine-tutorial}:
\begin{CompactEnumerate}
\item $\alpha$ is monotonic, \ie $x \leqconc
y \implies \alpha(x) \leqabst \alpha(y)$
\item $\gamma$ is monotonic, $a \leqconc
b \implies \gamma(a) \leqabst \gamma(b)$
\item $\gamma \circ \alpha$ is extensive, \ie $\forall c \in
\mathbb{C}: c \leqconc \gamma(\alpha(c))$
\item $\alpha \circ \gamma$ is reductive, \ie $\forall
a \in \mathbb{A}: \alpha(\gamma(a)) \leqabst a$
\end{CompactEnumerate}
The Galois connection is denoted as
\GaliosConnectionGeneric. The existence of a Galois
connection enables reasoning about the soundness and the
precision of any abstract operator. 

\Para{Optimality.} Suppose \GaliosConnectionGeneric is a
Galois connection. Given a concrete operator $f: \mathbb{C}
\rightarrow \mathbb{C}$, the abstract operator $\alpha \circ
f \circ \gamma$ is the smallest sound abstraction of $f$:
that is, for any sound abstraction $g : \bb{A}{} \rightarrow
\bb{A}{}$ of $f$, we have $\forall a \in \bb{A}{}:
\alpha(f(\gamma(a))) \leqabst g(a) $. We call $\alpha \circ
f \circ \gamma$ the {\em optimal}, or maximally precise
abstraction, of $f$.

\subsection{The Tnum Abstract Domain}

Tnums enable performing bit-level analysis by abstracting
each bit of a program variable separately.  Across
executions, each bit is either known to be 0, known to be 1,
or uncertain, denoted by $\mu$. For an \nbit program
variable, the abstract value corresponding to the variable
has $n$ ternary digits, or {\em trits}. Each trit has a
value of 0, 1, or $\mu$.

Bit-level abstract interpretation has been addressed in several prior
works using the bitfield abstract
domain~\cite{monniaux-device-driver-verification, mine-bitfield,
  regehr-deriving-abstract-transfer} and the known bits abstract
domain~\cite{dataflow-pruning-oopsla20, testing-static-analysis-cgo20,
  llvm-known-bits-analysis}. Abstraction and concretization functions
forming a Galois connection already exist~\cite{mine-bitfield}, as
well as sound and optimal abstract operators for bit-level operations
like bit-wise-and (\bitwiseAnd), bit-wise-or (\bitwiseOr), and shifts
(\bitwiseLshift, \bitwiseRshift)~\cite{testing-static-analysis-cgo20,
  mine-bitfield}. In contrast to prior work, this paper explores
provably sound, optimal, and computationally-efficient abstract
operators corresponding to {\em arithmetic operations} such as
addition, subtraction, and multiplication. The Linux kernel analyzer,
despite heavily leveraging this domain's abstract operations, formally
lays out neither the soundness nor optimality for the abstract
arithmetic operations.

\Para{Abstract and Concrete Domains.} Tnums track each bit
of variables drawn from the set of \nbit integers
$\mathbb{Z}_n$.
\begin{CompactItemize}
\item The concrete poset is $\mathbb{C} \triangleq
2^{\mathbb{Z}_n}$, the power set of $\mathbb{Z}_n$. The
ordering relation $\leqconc$ is the subset relation:

\begin{small}
\begin{align}
\label{eqn_concrete_ordering_relation}
\begin{split}
a \leqconc b \; \triangleq a \subseteq b
\end{split}
\end{align}
\end{small}
\item The abstract poset $\mathbb{A}$ is the set of \ntrit
tnums $\mathbb{T}_{n}$ (each trit is 0, 1, or $\mu$).
Suppose we represent the trit in the $\nth{i}$ position of
$a$ by $\ithbitx{a}{i}$. The ordering relation $\leqabst$
between abstract elements is defined by:

\begin{small}
\begin{align}
\label{eqn_tnum_ordering_relation}
\begin{split}
P \leqabst Q \; \triangleq & \; \forall i, 0 \leq i 
\leq n-1, \forall k \in \{0, 1\}: \\ 
& (\ithbitx{P}{i} = \mu \Rightarrow \ithbitx{Q}{i} = \mu) 
\wedge (\ithbitx{Q}{i} = k \Rightarrow \ithbitx{P}{i} = k)
\end{split}
\end{align}
\end{small}
\end{CompactItemize}

\Fig{lattices} shows Hasse diagrams of the lattices induced
by these posets for integers with bit width $n=2$. The
concrete domain consists of all elements of the power set of
$\{0, 1, 2, 3\}$ and the abstract domain consists of tnums
of the form $t_1 t_0$ where each $t_i$ is a trit with value
0, 1, or $\mu$. Any of $3^n$ abstract values can be
used to represent concrete sets of \nbit values.

\begin{figure}
\centering
\includegraphics[width=\columnwidth]{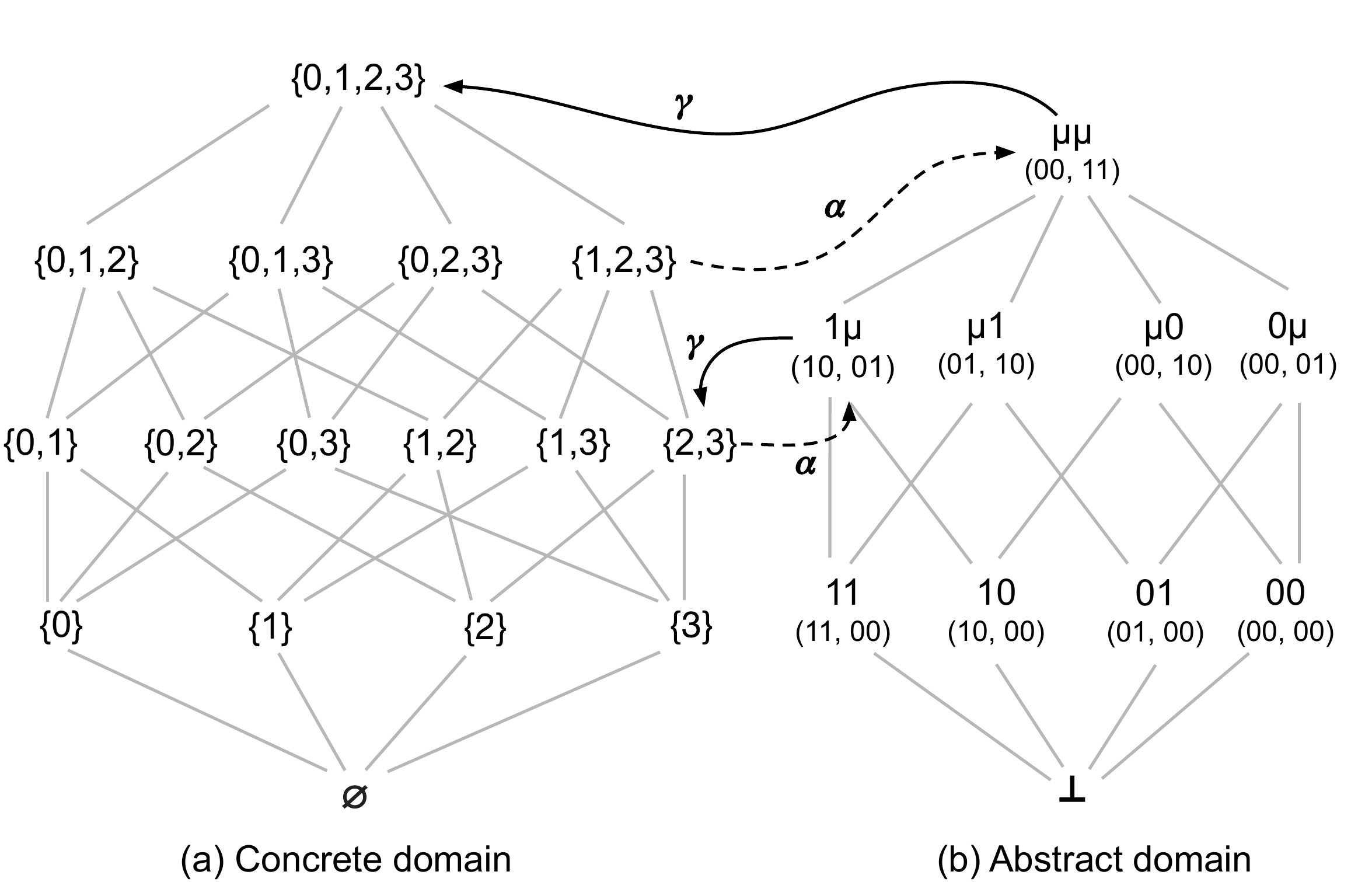}
\caption{\footnotesize Hasse diagrams of the lattices for (a) the
concrete domain $(2^{\bb{Z}{n}}, \subseteq)$ and (b) the
abstract domain $(\bb{T}{n}, \leqabst)$ for $n=2$. Below
every element of the abstract domain, we show its Linux
kernel implementation using two \ibit{2} values $(v, m)$.
Also shown are two examples of abstraction ($\alpha$,
dotted black lines) followed by concretization ($\gamma$,
solid black lines). (i) Starting with $C' = \{1, 2, 3 \}$,
$\alpha(C')$ gives $\mu\mu$, and $\gamma(\alpha(C'))$ gives
$\{0, 1, 2, 3\}$, an overapproximation of $C'$. (ii)
However, starting with $C'' = \{2, 3 \}$, $\alpha(C'')$
gives $1\mu$, and $\gamma(\alpha(C''))$ gives $\{2, 3\}$,
exactly equal to $C''$. In both cases, $C \leqconc
\gamma(\alpha(C))$.}
\label{fig:lattices}
\end{figure}

\Para{Implementation of tnums in the Linux kernel.} The
Linux kernel's implementation of representing one \ntrit
tnum $P \in \mathbb{T}_{n}$ uses two \nbit values
$\tnumi{P}$, where the `$v$' stands for {\em value} and the
`$m$' stands for {\em mask.} The values of the $\nth{k}$
bits of $\tvi{P}$ and $\tmi{P}$ are used to inform the value
of the $\nth{k}$ trit of $P$.

\begin{small}
\begin{align}
\label{eqn_tnumbasic}
\begin{split}
    (\ithbitx{\tvi{P}}{k} = 0 \wedge \ithbitx{\tmi{P}}{k} = 0) \;\; \triangleq \;\;  \ithbitx{P}{k} = 0 \\
    (\ithbitx{\tvi{P}}{k} = 1 \wedge \ithbitx{\tmi{P}}{k} = 0) \;\; \triangleq \;\;  \ithbitx{P}{k} = 1 \\
    (\ithbitx{\tvi{P}}{k} = 0 \wedge \ithbitx{\tmi{P}}{k} = 1) \;\; \triangleq \;\;  \ithbitx{P}{k} = \mu
\end{split}
\end{align}
\end{small}

We define the domain of abstract values $\mathbb{T}_{n}
\triangleq \mathbb{Z}_n \times \mathbb{Z}_n$. If for a tnum
$P$, $\ithbitx{\tvi{P}}{k} = \ithbitx{\tmi{P}}{k} = 1$ at
some position $k$, we say that such a tnum is not {\em
well-formed}. All such tnums represent the abstract value
$\bot$ and the concrete empty set $\varnothing$.

\begin{small}
\begin{align}
\label{eqn_bottom_wellformed}
\begin{split}
  \forall P: (\exists k: \;\; \ithbitx{\tvi{P}}{k} \;\;  \bitwiseAnd
  \;\; \ithbitx{\tmi{P}}{k} = 1) \Leftrightarrow \;\; P = \bot\\
\end{split}
\end{align}
\end{small}

A large fraction of random bit patterns $(v, m)$ aren't well
formed: in particular, only $3^n$ among the $2^{2n}$ \nbit
$(v,m)$ bit patterns correspond to well-formed tnums that
are not $\bot$. 

We are now ready to define the Galois connection for the
tnum abstract domain using the above implementation of
abstract values. These take a form similar to the functions
defined in prior work~\cite{mine-bitfield}. In the
discussion that follows we will use the notation
(\bitwiseAnd, \bitwiseOr, \bitwiseXor, \bitwiseNot,
\bitwiseLshift, \bitwiseRshift) respectively for the bitwise
and, or, exclusive-or, negation, left-shift, and right-shift
operations over \nbit bit vectors.

\Para{Galois connection.} Given a concrete set
$C \in 2^{\mathbb{Z}_n}$. The abstraction function $\alpha:
2^{\mathbb{Z}_n} \rightarrow \mathbb{Z}_n \times \mathbb{Z}_n$
is defined as follows.

\begin{align}
\begin{split}
\label{eqn_alpha}
\alpha_{\bitwiseAnd}(C) \; & \triangleq \; \bitwiseAnd  \big\{ c \mid c \in C \big\} \\
\alpha_{\bitwiseOr}(C) \; & \triangleq \; \bitwiseOr  \big\{ c \mid c \in C\big\} \\
\alpha(C) \; & \triangleq \; ( \; \alpha_{\bitwiseAnd}(C), \;\; \alpha_{\bitwiseAnd}(C) \; \bitwiseXor \; \alpha_{\bitwiseOr}(C) \; )
\end{split}
\end{align}
This abstraction function is sound. However, it is not
exact, as easily seen from the fact that there are $2^{2^n}$
elements in $\mathbb{C}$ but only $3^n$ well-formed tnums in
$\mathbb{T}_{n}$. Many concrete sets will be
over-approximated. However, $\alpha$ is a composition of
functions that abstract the domain exactly {\em when each
bit is considered
separately}~\cite{bitvector-dataflow-analysis-toplas94,
testing-static-analysis-cgo20}. Informally, given a concrete
set $C \in \bb{C}{}$ and $x, y \in C$, $\alpha(C)$ contains
an uncertain trit at position $k$ iff $C$ contains $x$ and
$y$ with bits differing at $k$.

\begin{align}
\begin{split}
\label{eqn_bitwise_exact}
\forall b \in \{0, 1\}: \;\; \ithbitx{\alpha(C)}{k} = b \;\;
\Leftrightarrow \;\;
\forall x \in C : \;\; \ithbitx{x}{k} = b \\
\ithbitx{\alpha(C)}{k} = \mu \;\; \Leftrightarrow \;\;
\exists x, y \in C : \;\; \ithbitx{x}{k} = 0 \wedge
\ithbitx{y}{k} = 1
\end{split}
\end{align}
This abstraction function $\alpha$ is {\em bitwise exact}.

Further, consider a tnum $P \in \bb{T}{n}$ implemented as
$\tnumi{P} \in \mathbb{Z}_n \times \mathbb{Z}_n.$ Then the
concretization function $\gamma: \mathbb{Z}_n \times
\mathbb{Z}_n \rightarrow 2^{\mathbb{Z}_n}$ is defined as:

\begin{align}
\begin{split}
\label{eqn_gamma}
\gamma(P) \; = \; \gamma(\tnumi{P}) \; & \triangleq \; \big\{ c \in \mathbb{Z} \mid c \; \bitwiseAnd \; \bitwiseNot \tmi{P} = \tvi{P} \big\} \\
\gamma(\bot) \; & \triangleq \; \varnothing
\end{split}
\end{align}
Then $\alpha$ and $\gamma$ form a Galois connection.
Informally, the tnum obtained from applying $\alpha$ on a
set of concrete values always soundly over-approximates the
original set if concretized. An illustration of this fact
can be seen in \Fig{lattices}. Please refer to the
\Supplement{}~\cite{vishwanathan-tnums-arxiv} for the
(standard) proof. The existence of the Galois connection
enables, in principle, constructing sound and optimal
abstract operators over tnums. The abstraction of the
concrete set $\{1, 2, 3\}$ soundly overapproximates it:
$\{1, 2, 3\} \leqconc \gamma(\alpha(\{1, 2, 3\}))$.

\Para{Abstract operators on tnums.} The BPF instruction set
supports the following (typical) concrete operations over
64-bit registers: \texttt{add, sub, mul, div, or, and, lsh,
rsh, neg, mod, xor} and \texttt{arsh}. To soundly analyze
general BPF programs, the BPF static analyzer requires
abstract operators corresponding to all the above concrete
operations. For some operators, notably \texttt{div} and
\texttt{mod}, defining a precise abstract operator is
challenging. In such cases, the BPF static analyzer
conservatively and soundly sets all the output trits to
unknown.

\Para{Challenges.}
Despite enjoying a Galois connection, constructing {\em
  efficient} optimal abstractions for arithmetic
operators is non-trivial. Given a concrete operator $f$, the
optimal abstract operator $\alpha \circ f \circ \gamma$ is
infeasible to compute in
practice~\cite{symbolic-best-transformer-vmcai04,
  symbolic-computation-of-abstraction-operations-cav12,
  posthat-and-all-that-tapas13}. For example, if $f$ is a
concrete operator of arity 2, there may be $2^{2n}$
computations of $f$ after the first concretization
$\gamma(.)$ in the worst case (the average case is not much
better).

Prior work on the bitfield domain~\cite{mine-bitfield}, a domain
similar to tnums \GaliosConnectionIntegerToBitField, presents abstract
operators for bitwise or, and, exclusive-or, left and right shift
operations that are \emph{optimal}. However, most prior works on the
bitfield and known bits abstract
domains~\cite{dataflow-pruning-oopsla20,
  testing-static-analysis-cgo20, llvm-known-bits-analysis,
  monniaux-device-driver-verification, mine-bitfield} fail to provide
abstract arithmetic operators for addition,
subtraction, and multiplication.  To our knowledge, Regehr and
Duongsaa~\cite{regehr-deriving-abstract-transfer} provide the only
known abstract operators for arithmetic in this domain, based on
ripple-carry logic and composition of abstract operators. These
operators are sound but not optimal. Further, they have a runtime of
$O(n)$ for \nbit abstract addition and subtraction, and $O(n^2)$ for
abstract multiplication.

In the next section, we present proofs of soundness and
optimality for abstract operators for addition and
subtraction originally developed (without formal proof) in
the Linux kernel. These operators run in $O(1)$ time given
\nbit machine arithmetic instructions ($n=64$ in the
kernel). Such efficiency is remarkable, given that in
general addition and subtraction use ripple-carry operations
creating dependencies between the bits. We also present an
abstract multiplication operator that is provably sound,
empirically more precise, and faster than the abstract
multiplication in \cite{regehr-deriving-abstract-transfer}
and the Linux kernel. Notably, none of the algorithms in
this paper use the composition structure $\alpha \circ f
\circ \gamma$ or ``merely'' compose existing sound abstract
operators. This motivated us to develop dedicated proof
techniques.

\section{Soundness and Optimality of Abstract Arithmetic
  over Tnums}
\label{sec:verification}

We explore the soundness and optimality of tnum arithmetic operators,
specifically addition, subtraction, and multiplication. 
The kernel proposes abstract operators for each of them, but lacks any
proof of soundness.  Hence, we perform an automated (bounded bitwidth)
verification of the soundness of the kernel's tnum abstract operators
(\Sec{automated-verification}) using SMT solvers. We were able to
prove the soundness of the kernel's abstract addition, subtraction,
and all other bitwise operators up to 64-bits, and soundness of the
kernel's multiplication up to 8-bits.  Motivated by these results, we
undertook an analytical study of these algorithms, which led us to
paper-and-pen proofs of both {\em soundness and optimality} of the
kernel's abstract operators for addition and subtraction over {\em
  unbounded bitwidths} (\Sec{unbounded-bitwidth-proofs}). We were
unable to analytically prove the soundness of the kernel's tnum
multiplication for unbounded bitwidths. Hence, we developed a new
algorithm for tnum multiplication that is provably sound for unbounded
bitwidths, and empirically more precise and faster than all prior
implementations (\Sec{alternative-tnum-multiplication}).

\subsection{Automatic Bounded Verification of Kernel Tnum
  Arithmetic} 
\label{sec:automated-verification}
We encode verification conditions corresponding to the soundness of
tnum abstract arithmetic operators in first order logic and discharge
them to a solver. We use the theory of bitvectors. Our verification
conditions are specific to a particular bitwidth ($n$).
We use \ibit{64}
bitvectors to encode the tnum operations wherever feasible ($n=64$ in
the kernel). For a tnum $P$ drawn from the set of \ntrit tnums
$\bb{T}{n}$, we denote its kernel implementation by $\tnumi{P} \in
\bb{Z}{n} \times \bb{Z}{n}$.

\Para{Soundness of 2-ary operators.} Recall from Section
\Sec{background} the notion of soundness of an abstract
operator.  We can generalize this notion to 2-ary operators
$\opT: \bb{T}{n} \times \bb{T}{n} \rightarrow \bb{T}{n}$ and
$\opC: \bb{Z}{n} \times \bb{Z}{n} \rightarrow \bb{Z}{n}$. We
say that $\opT$ is a sound abstraction of $\opC$ iff the
following condition (\Eqn{eqn_soundness}) holds.

\begin{small}
\begin{align}
\begin{split}
\label{eqn_soundness} 
&\forall P, Q \in \bb{T}{n}: \\
& \Big\{  \opC(x, y) \mid x \in \gamma(P), 
y \in \gamma(Q) \Big\}
\;  \leqconc \;
\gamma(\opT(P, Q))
\end{split}
\end{align}
\end{small}

To encode \eqref{eqn_soundness} in first-order logic,
recall that the concrete order $\leqconc$ is just the subset
relationship between the two sets $\subseteq$. At a high
level, the subset relationship $S_1 \subseteq S_2$ in
\eqref{eqn_soundness} can be encoded by universally
quantifying over the members of $S_1$ and writing down the
query $\forall x \in \bb{Z}{n}: x \in S_1 \Rightarrow x \in
S_2$. The formula $x \in S_1$ is easy to encode given the
left-hand side of \eqref{eqn_soundness}.
To encode $x \in S_2$ from the right-hand side of
\eqref{eqn_soundness}, we define a {\em membership
predicate}. This predicate asserts that $x \in \gamma(R)$
where $R \triangleq \opT(P, Q)$.
Finally, we ensure that the universally quantified tnums $P$
and $Q$ are non-empty, and encode the action of the concrete
and abstract operators $\opC$ and $\opT$ in logic. The
details follow.

\Para{Membership predicate $x \in \gamma(P).$} Consider a
concrete value $x$ that is contained in the concretization
of tnum $P$. Using the definition of the concretization
function in \eqref{eqn_gamma}, we write down the predicate
$\memberPredicate{}$: 

\begin{small}
\begin{align}
\label{eqn_mem}  
\memberPredicate(x, P) \; \triangleq \; x \; \bitwiseAnd \; \bitwiseNot \tmi{P} = \tvi{P}
\end{align}
\end{small}
\Para{Quantifying over well-formed tnums.} To ensure that
\eqref{eqn_soundness} only quantifies over non-empty tnums,
we encode one more predicate, $\wellformedPredicate$, based
on \eqref{eqn_bottom_wellformed}:

\begin{small}
\begin{align}
\label{eqn_wff}
\wellformedPredicate(P) \triangleq \, \tvi{P} \, \bitwiseAnd \, \tmi{P} = 0
\end{align}
\end{small}
\Para{Putting it all together.} The soundness
predicate for a given pair of abstract and concrete
operators $\opT, \opC$ is

\begin{small}
\begin{align}
\begin{split}
\label{eqn_soundness_2} 
\forall & P, Q \in \bb{T}{n}, x, y \in \bb{Z}{n} \; : 
\\ & \wellformedPredicate(P) \; \wedge \wellformedPredicate(Q) \; \wedge \memberPredicate(x, P) 
\\ & \; \wedge \memberPredicate(y, Q) \; \wedge z = \opC(x, y) \; \wedge R = \opT(P, Q) \;
\\ & \Rightarrow \memberPredicate(z, R)
\end{split}
\end{align}
\end{small}

An SMT solver can show the validity of this formula by
proving that the negation of this formula is unsatisfiable. 

\Para{Example: encoding abstract tnum addition.} 
We show how to encode the soundness of the abstract addition
operator over tnums.  The kernel uses the algorithm
\code{tnum_add} from \Lst{tnum_add} to perform abstract
addition over two tnums. The predicate $\addPredicate$ below
captures the result of abstract addition of $P$ and $Q$ into
$R$.

\begin{small}
\begin{align}
\begin{split}
\label{eqn_add_predicate}
add&( P, Q, R) \triangleq \; \\
   & (sv = \tvi{P} \, \arithAdd \, \tvi{Q}) 
      \land (sm = \tmi{P} \,\arithAdd \, \tmi{Q})  \land (\Sigma = sv \,\arithAdd \, sm) \\
   &  \land (\chi = \Sigma \, \bitwiseXor \, sv) \land (\eta = \chi \, \bitwiseOr \, \tmi{P} \, \bitwiseOr \,\tmi{Q} ) 
      \land (\tvi{R} = sv  \, \bitwiseAnd \, \bitwiseNot \, \eta) \\
   & \land (\tmi{R} = \eta) 
\end{split} 
\end{align}
\end{small}

We can plug in the $\addPredicate$ predicate in place of
$\opT$ in \Eqn{eqn_soundness_2}. The function $\opC$ is just
\nbit bitvector addition.

\begin{lstlisting}[
caption={ {\footnotesize Linux kernel's implementation of tnum addition 
(\code{tnum_add})}}, label={lst:tnum_add}]
def tnum_add(tnum P, tnum Q):

	u64 sv := P.v + Q.v
	u64 sm := P.m + Q.m
	u64 sigma := sv + sm
	u64 chi := sigma (*$\bitwiseXor$*) sv
	u64 eta := chi | P.m | Q.m	
	tnum R := tnum(s(*$\textsubscript{v}$*) & ~eta, eta)	
	return R
\end{lstlisting}
\begin{lstlisting}[caption={ {\footnotesize Linux 
kernel's implementation of tnum multiplication 
(\kernmul)}}, label={lst:kern_mul}]
def  $\kernmul$(tnum P, tnum Q)

	tnum pi := tnum(P.v * Q.v, 0)
	tnum ACC := hma(pi, P.m, Q.m | Q.v)
	tnum R:= hma(ACC, Q.m, P.v)
	return R
	
def hma(tnum ACC, u64 x, u64 y)

	while (y):
		if (y(*$\textsubscript{[0]}$*) == 1)
			ACC := tnum_add(ACC, tnum(0, x))
		y := y >> 1
		x := x << 1
	return ACC
\end{lstlisting}

\Para{Observations from bounded verification.} We encoded
the first-order logic formulas to perform bounded
verification of the soundness of the following tnum
operators defined in the Linux kernel: addition,
subtraction, multiplication, bitwise or, bitwise and,
bitwise exclusive-or, left-shift, right-shift, and
arithmetic right-shift. We have spot-checked the correctness
of our encodings with respect to the kernel source code
using randomly-drawn tnum inputs; the details of this
testing harness are in our
\Supplement{}~\cite{vishwanathan-tnums-arxiv}.

For all operators except multiplication, verification
  succeeded for bitvectors of width $64$ in just a few
  seconds. In contrast, verification of multiplication
  (\kernmul), shown in \Lst{kern_mul}, succeeds quickly at
  bitwidth $n=8$, but does not complete even after 24 hours
  with bitwidth $n=16$. This is due to the presence of
  non-linear operations and large unrolled loops. This
  observation motivated us to develop a new, provably sound
  algorithm for tnum multiplication
  (\Sec{alternative-tnum-multiplication}).

Further, our bounded verification efforts helped us uncover
non-obvious properties of tnum arithmetic: (1) tnum
addition is not associative, (2) tnum addition and
subtraction are not inverse operations, and (3) tnum
multiplication is not commutative.

\subsection{Soundness and Optimality of Tnum Abstract Addition}
\label{sec:unbounded-bitwidth-proofs}
\label{sec:tnum-add-proof-maintext}

We present an analytical proof of the soundness and
optimality of the kernel's abstract addition operator for
unbounded bitwidths.  The proof for subtraction, which is
very similar in structure, is in our
\Supplement{}~\cite{vishwanathan-tnums-arxiv}.

\begin{figure}[t]
  \centerline{\includegraphics[width=\columnwidth]{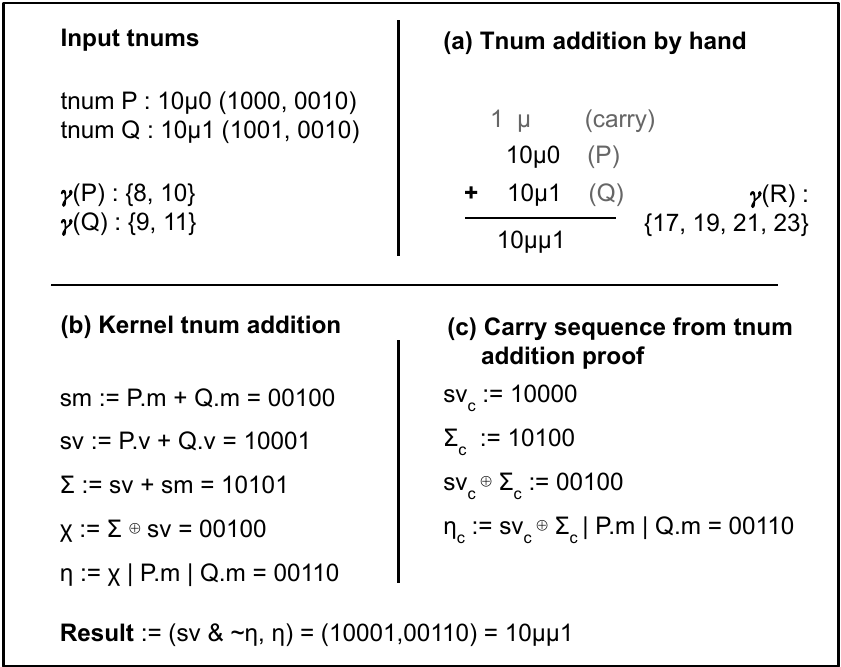}}
  \caption{\footnotesize Illustration of tnum addition. We provide a side by side comparison of (a) tnum addition by hand and (b) the kernel algorithm for tnum addition as well as (c) the carry sequence in the operation as discussed in the proof of tnum addition.}
  \label{fig:our_tnum_add}
\end{figure}

\Para{An example.} The source code for abstract addition
(\code{tnum_add}) is shown in \Lst{tnum_add}.
Figure~\ref{fig:our_tnum_add} illustrates tnum addition with
an example. In particular, adding two tnums ``by hand'', as
shown in \Fig{our_tnum_add}(a), propagates uncertainty
explicitly in the carries, rippling the carry bits through
the tnums one bit position at a time. However, as seen in
\Fig{our_tnum_add}(b), \code{tnum_add} does not use any such
ripple-carry structure in its computations. Yet, as we show
later (and illustrated in \Fig{our_tnum_add}(c)),
\code{tnum_add} implicitly reasons about the unknown bits in
the sequence of carries produced during the addition.

\begin{definition}
\label{addition-of-bits}
\textbf{Full adder equations.}
When adding two concrete binary numbers $p$ and $q$, each
bit of the addition result $r$ is set according to the
following:
$$r[i] = p[i] \,\bitwiseXor\, q[i] \,\bitwiseXor\,
\ithbit{c}{{in}}{i}$$ where $\bitwiseXor$ is the
exclusive-or operation and $\ithbit{c}{{in}}{i} =
\ithbit{c}{{out}}{i-1}$ and $\ithbit{c}{{out}}{i-1}$ is the
carry-out from the addition in bit position $i-1$. The
carry-out bit at the $i^{th}$ position is given by
$$\ithbit{c}{{out}}{i} = (p[i] \,\bitwiseAnd\, q[i])
\,\bitwiseOr\, (\ithbit{c}{{in}}{i} \,\bitwiseAnd\, (p[i]
\,\bitwiseXor\, q[i]))$$
\end{definition}

\Para{Key proof technique.} We show the soundness and
optimality of \code{tnum_add} by reasoning about the set of
all possible {\em concrete outputs}, \ie the results of
executions of concrete additions over elements of the input
tnums $P, Q \in \bb{T}{}$.
If we denote by $+$ the concrete addition operator over
$\bb{Z}{n}$, this is the set $\{p + q \mid p \in \gamma(P)
\wedge q \in \gamma(Q)\}$ or $+(\gamma(P), \gamma(Q))$ in
short.
The proof proceeds by finding bit positions in the concrete
output set that can be shown to be either a 1 or a 0 in {\em
  all members} of that set (respectively lemmas
\ref{lemma:minimum-carries-lemma} and
\ref{lemma:maximum-carries-lemma}). Every other bit position
is such that there are elements in the concrete output set
that differ at that bit position.
\Lemma{lemma:capture-uncertainty-lemma} invokes the bitwise-exactness
(\Eqn{eqn_bitwise_exact}) of the abstraction function $\alpha$, and
along with \Lemma{thm:add-soundness-maximal-precision}, shows that
\code{tnum_add} is a sound and optimal abstraction for $+$
(\ie the same as $\alpha\ \circ\ +\ \circ\ \gamma$).

Consider the addition that occurs ``by hand'' in
\Fig{our_tnum_add}(a). Intuitively, at a given bit position
of the output tnum, the result will be unknown if either of
the operand bits $\ithbitx{p}{i}$ or $\ithbitx{q}{i}$ is
unknown, or if the carry-in bit $\ithbitx{c_{in}}{i}$
(generated from less-significant bit positions) is
unknown. Note that these three bits may be (un)known
independent of each other since they depend on different
parts of the input tnums. The crux of the proof lies in
identifying which carry-in bit positions vary across
different concrete additions.  This is done by
distinguishing the carries generated due to the unknown bits
in the operands from the carries that will be present or
absent in {\em any} concrete addition drawn from the input
tnums. In the example in \Fig{our_tnum_add}(a), the sequence
of carries is $10\mu00$, with the middle carry-in bit being
uncertain and all others known to be 0s or 1s in all
concrete additions from the input tnums.

Suppose $p$ and $q$ are two concrete values in tnum $P$
and tnum $Q$, respectively, \ie $p \in \gamma(P), q \in
\gamma(Q).$

\begin{lemma}
\label{lemma:minimum-carries-lemma}
\textbf{Minimum carries lemma.} The addition $sv = P.v +
Q.v$ will produce a sequence of carry bits that has the
least number of 1s out of all possible additions $p +
q$.
\end{lemma}
The consequence of this lemma is that {\em any} concrete addition $p$
+ $q$ will produce a sequence of carry bits with 1s in at least those
positions where the $sv$ addition produced carry bits set to 1
(the \Supplement{} provides a proof of this lemma).
\Fig{our_tnum_add}(c) shows
the set of carries produced in $sv$ (\ie $sv_c \triangleq 10000$). Any
addition $p + q$ will produce a 1-bit carry in the same positions as
the 1 bits in $sv_c$.

\begin{lemma}
\label{lemma:maximum-carries-lemma}
\textbf{Maximum carries lemma.} The addition $\Sigma = (P.v + P.m) + (Q.v +
Q.m)$ will produce the sequence of carry bits with the most number of 1s out of
all possible additions $p$ + $q$. 
\end{lemma}

The consequence of this lemma is that {\em any} concrete addition $p$
+ $q$ will produce a sequence of carry bits with 0s in at least those
positions where the $\Sigma$ addition produced carry bits set to 0
(proof is available in our \Supplement{}).
\Fig{our_tnum_add}(c) shows the set of carries produced in $\Sigma$
(\ie $\Sigma_c \triangleq 10100$). Any addition $p + q$ will produce a
0-bit carry in the same positions as the 0 bits in $\Sigma_c$.

\begin{lemma}
\label{lemma:capture-uncertainty-lemma}
\textbf{Capture uncertainty lemma.}  Let $sv_{c}$ and
$\Sigma_{c}$ be the sequence of carry-in bits from the
additions in $sv$ and $\Sigma$, respectively. Suppose
$\chi_c \triangleq sv_{c} \,\bitwiseXor\, \Sigma_{c}$. The
bit positions $k$ where $\ithbitx{\chi_c}{k} = 0$ have carry
bits fixed in all concrete additions $p+q$ from
$+(\gamma(P), \gamma(Q))$. The bit positions $k$ where
$\ithbitx{\chi_c}{k} = 1$ vary depending on the concrete
addition: \ie $\exists p_1, p_2 \in \gamma(P), q_1, q_2 \in
\gamma(Q)$ such that $p_1 + q_1$ has its carry bit set at
position $k$ but $p_2 + q_2$ has that bit unset.
\end{lemma}

Intuitively, from the minimum carries lemma, any carry bit that is set
in $sv_c$ must be set in the sequence of carry bits in any concrete
addition $p + q$.  Similarly, from the maximum carries lemma, any
carry bit that is unset in $\Sigma_c$ must be unset in the sequence of
carry bits in any concrete addition $p + q$. Hence, $sv_{c}
\,\bitwiseXor\, \Sigma_{c}$ represents the carries that may arise
purely from the uncertainty in the concrete operands picked from $P$
and $Q$. Further, these carries do in fact differ in two concrete
additions $sv$ and $\Sigma$. From the bitwise-exactness of the tnum
abstraction function $\alpha$ (\Eqn{eqn_bitwise_exact}), it follows
that these are precisely the bits that must be unknown in the
resulting tnum due to the carries. See the \Supplement{} for a
detailed proof.

Hence, the mask in the resulting tnum must be $(sv_c\
\bitwiseXor\ \Sigma_c) \,\bitwiseOr\, P.m \,\bitwiseOr\,
Q.m$. However, \code{tnum_add} uses the final mask $(sv\
\bitwiseXor\ \Sigma) \,\bitwiseOr\, P.m \,\bitwiseOr\, Q.m$
(see \Lst{tnum_add}).
\Lemma{thm:add-soundness-maximal-precision} shows that these
two quantities are, in fact, always the same.

\begin{lemma}
\label{thm:add-soundness-maximal-precision}
\textbf{Equivalence of mask expressions.}
The expressions $(sv
\,\bitwiseXor\, \Sigma) \,\bitwiseOr\, P.m \,\bitwiseOr\,
Q.m$ and $(sv_{c} \,\bitwiseXor\, \Sigma_{c}) \,\bitwiseOr\,
P.m \,\bitwiseOr\, Q.m$ compute the same result.
\end{lemma}

We prove this lemma using the rules of propositional logic
in our \Supplement{}. Together, these lemmas allow us to
show the soundness and optimality of \code{tnum_add} below.

\begin{theorem}
\label{thm:soundness-optimality-of-tnum-add}
\textbf{Soundness and optimality of tnum\_add} The algorithm
\code{tnum_add} shown in \Lst{tnum_add} is a sound and
optimal abstraction of concrete addition over \nbit
bitvectors for unbounded $n$.
\end{theorem}

\subsection{Sound and Efficient Tnum Abstract Multiplication}
\label{sec:alternative-tnum-multiplication}

This section describes a novel algorithm for tnum
multiplication and a proof that it is a sound abstraction of
multiplication of \nbit concrete values for unbounded
$n$. Our algorithm has $O(n)$ run time. It is not an optimal
abstraction of concrete multiplication. However, as we show
later (\Sec{evaluation}), our algorithm is empirically more
precise and faster than all known prior implementations of
multiplication in this abstract domain. We were able to
contribute our algorithm to the tnum implementation in the
latest Linux kernel. 

\begin{lstlisting}[caption={\footnotesize A simplified implementation of
  our tnum multiplication algorithm
(\ourmulsimplified).}, label={lst:our_mul_simplified}]
def $\ourmulsimplified$(tnum P, tnum Q):

	ACC(*$\textsubscript{v}$*) := tnum(0, 0)
	ACC(*$\textsubscript{m}$*) := tnum(0, 0)

	# loop runs bitwidth times
	for i in range(0, bitwidth):
		# LSB of tnum P is a certain 1
		if (P.v(*$\textsubscript{[0]}$*) == 1) and (P.m(*$\textsubscript{[0]}$*) == 0):
		  ACC(*$\textsubscript{v}$*) := tnum_add(ACC(*$\textsubscript{v}$*), tnum(Q.v, 0))
		  ACC(*$\textsubscript{m}$*) := tnum_add(ACC(*$\textsubscript{m}$*), tnum(0, Q.m))
		# LSB of tnum P is uncertain
		else if  (P.m(*$\textsubscript{[0]}$*) == 1):
		  ACC(*$\textsubscript{m}$*) := tnum_add(ACC(*$\textsubscript{m}$*), tnum(0, Q.v|Q.m))
		# Note: no case for LSB is certain 0
		P := tnum_rshift(P, 1)
		Q := tnum_lshift(Q, 1)

	tnum R := tnum_add(ACC(*$\textsubscript{v}$*), ACC(*$\textsubscript{m}$*))
	return R
\end{lstlisting}

\Para{Our algorithm \ourmul through an example.}
\label{sec:our-mul-algorithm} Our tnum abstract
multiplication algorithm is shown in \Lst{our_mul}.  The
algorithm in \Lst{our_mul_simplified} is semantically
equivalent to it, but easier to understand, so we explain
the algorithm and its proof primarily using the algorithm in
\Lst{our_mul_simplified}.

Similar to the prior multiplication algorithms proposed in
bit-level reasoning
domains~\cite{regehr-deriving-abstract-transfer,
  tnum-kernel-source}, our algorithm is inspired by the long
multiplication method to generate the product of two binary
values.  The algorithm proceeds in a single loop iterating
over the bitwidth of the input tnums. \Ourmul takes two
input tnums $P$ and $Q$, and returns a result $R$.

\Fig{ourmul_ex}(a) shows an example. Suppose we are given
tnums $P = \mu01\ \tnumival{P}{001}{100}$ and $Q =
\mu10\ \tnumival{Q}{010}{100}$ to multiply. Two fully
concrete \nbit binary numbers may be multiplied in two
steps: (i) by computing the products of each bit in the
multiplier ($P$) with the multiplicand ($Q$), to generate
$n$ {\em partial products,} and (ii) adding the $n$ partial
products after appropriately bit-shifting them. To
generalize long multiplication to \nbit tnums which contain
unknown ($\mu$) trits, we add new rules: $0 * \mu = 0; 1 *
\mu = \mu;$ and $\mu * \mu = \mu$. Since the partial
products themselves contain unknown trits, the addition of
the partial products must occur through the abstract
addition operator \code{tnum_add}.

\begin{lstlisting}[caption={ \footnotesize Our final tnum 
multiplication algorithm
(\ourmul).}, label={lst:our_mul}]
def  $\ourmul$(tnum P, tnum Q):

	ACC(*$\textsubscript{v}$*) := tnum(P.v * Q.v, 0)
	ACC(*$\textsubscript{m}$*) := tnum(0, 0)

	while P.value or P.mask:
		# LSB of tnum P is a certain 1
		if (P.v(*$\textsubscript{[0]}$*) == 1) and (P.m(*$\textsubscript{[0]}$*) == 0):
		  ACC(*$\textsubscript{m}$*) := tnum_add(ACC(*$\textsubscript{m}$*), tnum(0, Q.m))
		# LSB of tnum P is uncertain
		else if  (P.m(*$\textsubscript{[0]}$*) == 1):
		  ACC(*$\textsubscript{m}$*) := tnum_add(ACC(*$\textsubscript{m}$*), tnum(0, Q.v|Q.m))
		# Note: no case for LSB is certain 0
		P := tnum_rshift(P, 1)
		Q := tnum_lshift(Q, 1)
  
	tnum R := tnum_add(ACC(*$\textsubscript{v}$*), ACC(*$\textsubscript{m}$*))
	return R
\end{lstlisting}

\begin{figure}[t]
  \centerline{\includegraphics[width=\columnwidth]{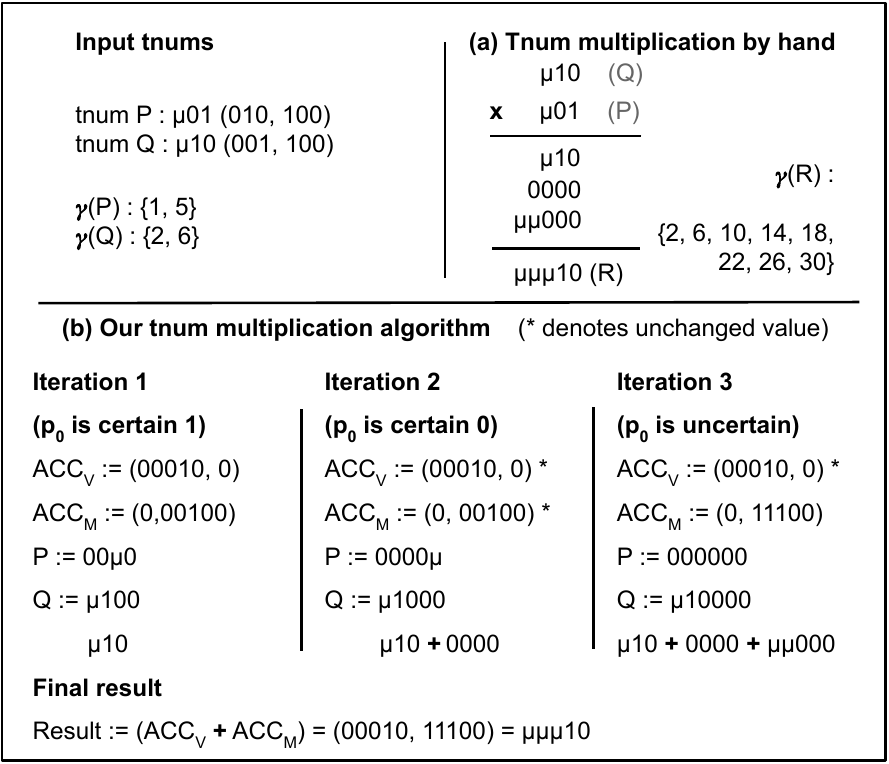}}
  \caption{\footnotesize Illustration of our tnum multiplication. We provide a side by 
  side comparison of (a) tnum multiplication by hand and (b) our improved 
  algorithm for tnum multiplication.}
  \label{fig:ourmul_ex}
 \end{figure}

Our tnum multiplication algorithm for the same pair of
inputs is illustrated in \Fig{ourmul_ex}(b). The algorithm
uses two tnums, $ACC_V$ and $ACC_M$, which are initialized
$(v, m) \triangleq (0,0)$. The tnums $ACC_V$ and $ACC_M$
accumulate abstract partial products generated in each
iteration using tnum abstract additions (\code{tnum_add}).
The algorithm proceeds as follows:
\begin{CompactEnumerate}
\item If the least significant bit of any concrete value $x
  \in \gamma(P)$ is known to be 1, then $ACC_V$
  (resp. $ACC_M$) accumulates the known bits (resp. unknown
  bits) in $Q$ (\eg iteration 1);
\item If the least significant bit of any $x \in \gamma(P)$
  is known to be 0, $ACC_V$ and $ACC_M$ remain unchanged
  (\eg iteration 2);
\item If the least significant bit of P is unknown ($\mu$),
  then $ACC_V$ is unchanged, but $ACC_M$ accumulates a tnum
  with a mask such that all possible bits that may be set in
  any $x \in \gamma(P)$ are also set in the mask (\eg 
  iteration 3).
\end{CompactEnumerate}
At the end of each iteration, $P$ (resp. $Q$) is bit-shifted
to the right (resp. left) by 1 position to ensure that the
next partial product is appropriately shifted before
addition.  The specific methods of updating $ACC_V$ and
$ACC_M$ in each iteration make \ourmul distinct from prior
multiplication
algorithms~\cite{regehr-deriving-abstract-transfer,
  tnum-kernel-source}. In particular, \ourmul decomposes the
accumulation of partial products into two tnums and uses
just a single loop over the bitwidth. These modifications
are crucial to the precision and efficiency of \ourmul
(\Sec{evaluation}).

\Para{Key proof techniques.}  Recall that a (unary) abstract
operator $g$ is a sound abstraction of a concrete operator
$f$ if $\forall a \in \bb{A}{}: f(\gamma(a)) \leqconc
\gamma(g(a))$. We show that our abstract multiplication
algorithm \ourmul is sound by showing that $\{x * y \mid x
\in \gamma(P) \wedge y \in \gamma(Q)\} \leqconc
\gamma(\ourmul(P, Q))$ for any tnums $P, Q \in \bb{T}{n}$.
We denote the former set $*(\gamma(P), \gamma(Q))$ in
short. The $*$ is the concrete multiplication over \nbit
bitvectors.

All the known abstract multiplication algorithms in this
domain are composed of abstract additions and abstract
shifts. A typical approach to prove soundness of such
operators is to invoke the result that when sound abstract
operators are composed soundly, \ie in the same way as the
corresponding concrete operators are composed, the result is
a sound abstraction of the composed concrete
operator~\cite[Theorem 2.6]{mine-tutorial}. The soundness of
the abstract multiplication from Regehr and
Duongsaa~\cite{regehr-deriving-abstract-transfer} may be
proved as a special case of this general result. However,
this approach is not applicable to proving the soundness of
\ourmul, since \ourmul's composition does not mirror any
composition of (concrete) additions and shifts to produce a
product.  Instead, we are forced to develop a proof
specifically for \ourmul by observing, through two
intermediate lemmas
(\Lemma{lemma:value-mask-decoupled-additions} and
\Lemma{lemma:union-tnum-with-zero}) that the concrete
products in $*(\gamma(P), \gamma(Q)) \in \gamma(\ourmul(P,
Q))$.

Below, we show a sketch of the proof of the soundness of
\ourmulsimplified, and argue
(\Lemma{lemma:strength-reduction-equivalence}) that \ourmul
is equivalent to \ourmulsimplified.

\begin{observation}
\label{obs:partial-products}
For two concrete bitvectors $x$ and $y$ of width $n$ bits,
the result of multiplication $y * x$ is just
$$ y * x = \sum_{k=0}^{n-1} \ithbitx{x}{k} * (y
\,\bitwiseLshift\, k)$$
\end{observation}

We call each term $\ithbitx{x}{k} \, *\, (y
\,\bitwiseLshift\,k)$ a {\em partial product.}

\begin{lemma}
\label{lemma:union-tnum-with-zero}
\textbf{Tnum set union with zero.} Given a non-empty tnum $P
\in \bb{Z}{n} \times \bb{Z}{n}$, define $Q \triangleq
tnum(0, P.v \mid P.m)$. Then, (i) $x \in \gamma(P)
\Rightarrow x \in \gamma(Q)$, and (ii) $0 \in \gamma(Q)$.
\end{lemma}

Intuitively, any tnum $(0, m)$ when concretized contains the
value 0.  Further, building a new tnum $Q$ whose mask has
set all the bits corresponding to the set value or mask bits
of $P$ ensures that $\gamma(P) \leqconc \gamma(Q)$. The full
proof is in our
\Supplement{}~\cite{vishwanathan-tnums-arxiv}. For example,
given $P = 0\mu1 = (001, 010)$ and $Q \triangleq (0, 011)$,
we have $\gamma(P) \leqconc \gamma(Q)$ and $0 \in
\gamma(Q)$.

For the next lemma, we define a variable-arity version of
\code{tnum_add} as follows:
$\code{tnum_add}_{j=0}^{n-1}(T_j)$ is evaluated by folding
the \code{tnum_add} operator over the list of tnum operands
$T_0, T_1, \cdots, T_{n-1}$ from left to right.

\begin{lemma}
\label{lemma:value-mask-decoupled-additions}
\textbf{Value-mask-decomposed tnum summations.} Given $n$
non-empty tnums $T_0, T_1, \ldots, T_{n-1} \in
\bb{T}{n}$. Suppose we pick $n$ values $z_0, z_1, z_2,
\ldots, z_{n-1} \in \bb{Z}{n}$ such that $\forall \, 0 \leq
j \leq n-1: \, z_j \in \gamma(T_j)$. Define tnum 

\begin{small}
\begin{align*}
\begin{split}
S \;\; \triangleq \;\;
\code{tnum_add}(\code{tnum_add}_{j=0}^{n-1} \;\;(tnum(T_j.v,
0)),
\\ \;\; \;\; \;\;\code{tnum_add}_{j=0}^{n-1} \;\;(tnum(0,
T_j.m)))
\end{split}
\end{align*}
\end{small}where $\code{tnum_add}_{(\cdot)}^{(\cdot)}$ is a
variable-arity version of \code{tnum_add} defined
above. Then, $\sum_{j=0}^{n-1} z_j \in \gamma(S)$.
\end{lemma}

Intuitively, suppose we had $n$ tnums $T_i, 0 \leq i \leq
n-1$ and we seek to construct a new tnum $S$ whose
concretization $\gamma(S)$ contains all possible {\em
concrete sums} from the $T_i$, \ie such that
$\{\sum_{j=0}^{n-1} x_j \mid x_i \in \gamma(T_i)\} \leqconc
\gamma(S)$. The most natural method to construct such a tnum
$S$ is to use the sound abstract addition operator
\code{tnum_add} over the $T_i$, \ie
$\code{tnum_add}_{j=0}^{n-1}(T_j)$. This lemma provides
another method of constructing such a tnum $S$: decompose
the tnums $T_i$ each into two tnums, consisting of the
values and the masks separately. Use \code{tnum_add} to
separately add the value tnums, add the mask tnums, and
finally add the two resulting tnums from the value-sum and
mask-sum to produce $S$. Then $S$ contains all concrete
sums. The full proof of this lemma is in the \Supplement{}.
For example, suppose $T_1 = 1\mu0 = (100, 010), T_2 = 01\mu
= (010, 001)$. Then $\forall x_1 \in \gamma(T_1), x_2 \in
\gamma(T_2): x_1 + x_2 \in \gamma(\code{tnum_add}((110, 0),
(0, 011)))$.

\begin{theorem}
\label{thm:our-mul-soundness}
\textbf{Soundness of \ourmul.} $\forall \, x \in \gamma(P)$,
$y \in \gamma(Q)$ the result $R$ returned by
\ourmulsimplified (\Lst{our_mul_simplified}) is such that $x
* y \in \gamma(R)$, assuming that abstract tnum addition
(\code{tnum_add}) and abstract tnum shifts
(\code{tnum_lshift}, \code{tnum_rshift}) are sound.
\end{theorem}

We prove this theorem by showing three properties, whose
full proofs are in the \Supplement{}. Below, $P_{in}$ and
$Q_{in}$ are the formal parameters to \ourmul.

\proofheading{Property \textnormal{P1}. $P$ and $Q$ are
  bit-shifted versions of $P_{in}$ and $Q_{in}$.} This
property follows naturally from the algorithm, which only
updates the tnums $P$ and $Q$ in the code using tnum
bit-shift operations (\code{tnum_lshift},
\code{tnum_rshift}). 

\proofheading{Property \textnormal{P2}. $ACC_V$ and $ACC_M$
  are value-mask-decomposed summations of partial products.}
There exist tnums $T_0, T_1, \ldots, T_{n-1}$ such that (i)
any concrete $\nth{j}$ partial product, $z_j \triangleq
\ithbitx{x}{j} \arithMul (y \bitwiseLshift j) \in
\gamma(T_j)$, for $0\leq j \leq n-1$; (ii) at the end of
the $\nth{k}$ iteration of the loop, $ACC_V =
\code{tnum_add}_{j=0}^{k-1} (tnum(T_j.v, 0))$, and (iii) at
the end of the $\nth{k}$ iteration of the loop, $ACC_M =
\code{tnum_add}_{j=0}^{k-1}(tnum(0, T_j.m))$.

At a high level, this property states that there is a set of
tnums $T_j$, where $\gamma(T_j)$ contains all possible
concrete values of the $\nth{j}$ partial product term $z_j
\triangleq \ithbitx{x}{j} \arithMul (y \bitwiseLshift j)$
(\Obs{partial-products}). In the example in
\Fig{ourmul_ex}(b), $T_0 = \mu10, T_1 = 0000, T_2 =
\mu\mu000$. In the case where $\ithbitx{P}{0}$ is $\mu$, we
use \Lemma{lemma:union-tnum-with-zero} to show that the
$T_j$ constructed by \ourmulsimplified is such that
$\gamma(Q) \leqconc \gamma(T_j)$ and $0 \in \gamma(T_j)$.
We also show that $ACC_V$ is the value-sum of the $T_j$ (see
\Lemma{lemma:value-mask-decoupled-additions}) while $ACC_M$
is the mask-sum. In \Fig{ourmul_ex}(b), $ACC_V \triangleq
\code{tnum_add}(010, 0000, 00000)$ and $ACC_M \triangleq
\code{tnum_add}(\mu00, 0000, \mu\mu000)$.

\proofheading{Property \textnormal{P3}. (Product
  containment) $\sum_{j=0}^{n-1}z_j \in
  \gamma(\code{tnum_add}(ACC_V, ACC_M)).$} 
That is, $\forall x \in \gamma(P), y \in \gamma(Q): x * y
\in \gamma(\code{tnum_add}(ACC_V, ACC_M))$.

This result follows from property P2 and
\Lemma{lemma:value-mask-decoupled-additions}. Property P3
concludes a proof of soundness of \ourmulsimplified:
$\forall x \in \gamma(P), y \in \gamma(Q): x * y \in
\gamma(R)$.

\begin{lemma}
\label{lemma:strength-reduction-equivalence}
\textbf{Correctness of strength reductions.}  \Ourmul
(\Lst{our_mul}) is equivalent to \ourmulsimplified
(\Lst{our_mul_simplified}) in terms of its input-output
behavior.
\end{lemma}

The existence of two accumulating tnums $ACC_V$ and $ACC_M$
in \ourmulsimplified allows us to use
\Lemma{lemma:value-mask-decoupled-additions} to prove
soundness. However, it is unnecessary to construct $ACC_V$
iteration by iteration. We observe that $ACC_V$ is merely
accumulating $(Q.v, 0)$ whenever $P[0]$ is known to be
1. All bits in each tnum accumulated into $ACC_V$ are known.
When \code{tnum_add} is used to add $n$ tnums $(v_i, 0), 0
\leq i \leq n-1$ it is easy to see that the result is
$(\sum_{i=0}^{n-1} v_i, 0)$. Using \Obs{partial-products},
we see that at end of the loop, $ACC_V = tnum(P.v * Q.v,
0)$. As an added optimization, \ourmul soundly terminates
the loop early if $P = (0, 0)$ at the beginning of any
iteration. 
Since \ourmul and \ourmulsimplified are equivalent, \ourmul
is also a sound abstraction of $*$.

While \ourmul is sound, it is not optimal. Key questions
remain in designing a sound and optimal algorithm with
$O(n)$ or faster run time.
(1) How can we incorporate correlation in unknown bits
across partial products? For example, multiplying $P = 11, Q
= \mu1$ produces the partial products $T_1 = 11, T_2 =
\mu\mu0$. However, the two $\mu$ trits in $T_2$ are
concretely either both 0 or both 1, resulting from the same
$\mu$ trit in $Q$. Failing to consider this in the addition
makes the result imprecise.
(2) Can we design a sound, precise, and fast
tnum addition operator of arity $n$?
(3) Eschewing long
multiplication, is it possible to use concrete
multiplication $(*)$ over tnum masks to determine the
unknown bits in the result?

\section{Experimental Evaluation}
\label{sec:evaluation}
Tnum operations are only one
component of the Linux kernel's BPF analyzer. 
To keep our measurement and contributions focused, our
evaluation focuses on the precision and speed
of our tnum multiplication operation relative to prior
algorithms.

\Para{Prior algorithms for abstract multiplication.}
Apart from Linux kernel's
multiplication~\cite{tnum-kernel-source}, Regehr and
Duongsaa \cite{regehr-deriving-abstract-transfer} also
provide an algorithm for abstract multiplication in a domain
similar to tnums, which is called the {\em bitwise} domain.
\Lst{regher_mul} presents their multiplication algorithm,
which we call \reghermul. We experimentally compare the
precision and performance of our tnum multiplication with
both \reghermul and the Linux kernel version (\kernmul). We
have performed bounded verification of the soundness of both
\kernmul and \reghermul at bitwidth $n=8$.

Similar to \ourmul, \reghermul is also inspired from long
multiplication. It generates partial products that are
subsequently added after appropriately bit-shifting them.
We observed that \reghermul needs to be carefully implemented with
machine arithmetic operations to have reasonable performance. In
\reghermul, the function \code{multiply_bit} modifies the second
operand $Q$ based on the $i^{th}$ trit of the first operand $P$.  If
the $i^{th}$ trit of $P$ is $\mu$, this modification is done in a
trit-by-trit fashion (\ie by iterating over the $n$ trits of $Q$ and
setting them to $\mu$).  To improve \reghermul's performance with
tnums, we substituted this loop with a single tnum operation
in our implementation: when $P$ is $\mu$, we construct a new
tnum \code{(0,
  Q.mask | Q.value)}, which has the same effect as individually setting
the trits of $Q$ to $\mu$.

\begin{lstlisting}[caption={
	{\footnotesize Bitwise multiplication (\ie \reghermul) by Regehr
	  \etal~\cite{regehr-deriving-abstract-transfer}}},
		label={lst:regher_mul}]
def $\reghermul$(tnum P, tnum Q):

	tnum sum = tnum(0, 0)
	# loop runs bitwidth times
	for i in range(0, bitwidth):
		tnum product = multiply_bit(P, Q, i)
		sum = tnum_add(sum, tnum_lshift(product, i))
	return sum

def multiply_bit(tnum P, tnum Q, u8 i):

	# Bit position i of tnum P is a certain 0
	if (P.v(*$_{[i]}$*) == 0 and P.m(*$_{[i]}$*) == 0):
		return tnum(0, 0)
	# Bit position i of tnum P is a certain 1
	elif (P.v(*$_{[i]}$*) == 1 and P.m(*$_{[i]}$*) == 0):
		return Q
	# Bit position i of tnum P is uncertain
	else: 
		# "kill" all bits of Q that are a certain 1,
		# i.e. set them to uncertain
		for j in range(0, bitwidth):
			if (Q.v(*$_{[j]}$*) == 1 and Q.m(*$_{[j]}$*) == 0) :
				Q.v(*$_{[j]}$*) = 0
				Q.m(*$_{[j]}$*) = 1
		return Q

\end{lstlisting}

\Para{Setup.} 
We performed all our experiments on the Cloudlab
\cite{cloudlab} testbed. We used two 10-core Intel Skylake
CPUs running at 2.20 GHz for a total of 20 cores, with 192GB
of memory.

\subsection{Evaluation of  Precision of \ourmul}
\label{sec:eval_precision}
We evaluate the precision of \code{our_mul} compared to
\reghermul and \kernmul by exhaustively evaluating all pairs
of tnum inputs at a given bitwidth $n$. We set $n=8$ to keep
the running times tractable.

Consider two abstract tnum multiplication operations \opone
and \optwo. Given two tnums $P$ and $Q$, suppose $R_1
\triangleq\opone(P, Q)$ and $R_2 \triangleq\optwo(P,
Q)$. For fixed $P$ and $Q$, \opone is more precise than
\optwo if $R_{1} \neq R_{2} \wedge R_{1} \leqabst R_{2}$, or
equivalently, $R_{1} \neq R_{2} \wedge \gamma(R_{1})
\leqconc \gamma(R_{2})$, where $\leqconc$ is the subset
relation $\subseteq$.
In general, two such output tnums $R_1$ and $R_2$ need not
be comparable using the abstract order $\leqabst$. For
example, at bitwidth $n=9$, with input tnums $P =
000000011$, $Q = 011\mu011\mu\mu$, the \kernmul algorithm
produces $R_1 = \mu\mu\mu\mu0\mu\mu\mu\mu$ whereas \ourmul
produces $R_2 = 0\mu\mu\mu\mu\mu\mu\mu\mu$. However,
empirically, for tnums of width $n=8$, outputs $R_1$ and
$R_2$ turn out to be always comparable using
$\leqabst$. That is, at $n=8$, $R_1 \neq R_2 \Rightarrow R_1
\leqabst R_2 \vee R_2 \leqabst R_1$. Hence, we can simply
compare the cardinalities of $\gamma({R_1})$ and
$\gamma({R_2})$ as a measure of the relative precision of
\opone and \optwo, for given input tnum pair $(P, Q)$.

Figure~\ref{fig-rel-pres} compares the cardinalities of
$\gamma(R_1)$ and $\gamma(R_2)$ when enumerating every
possible input tnum pair $(P, Q)$ of width $8$, and only
considering cases where $R_1 \neq R_2$ . Note that $R_1 \neq
R_2$ only when $R_1$ (similarly $R_2$) is a tnum with a
larger number of unknown trits ($\mu$) than $R_2$ (similarly
$R_1$).  We use a $log_{2}$ scale for the x-axis: each tick
on the x-axis to the right of $0$ is a point where \ourmul
produces a tnum that is more precise in exactly one trit
position when compared to the multiplication algorithm it is
pitted against. We observe that for around 80\% of the
cases, \ourmul produces a more precise tnum than both
\kernmul and \reghermul (the data to the right of 0 in
Figure \ref{fig-rel-pres}). In summary, our multiplication
is more precise than \kernmul and
\reghermul.

Note that \ourmul and \kernmul produce the same result $(R_1
= R_2)$ for
\percentageofourmulkernmulequalinputsatbitwidtheight\% of
all possible $8$-trit input tnum pairs. We evaluate how the
relative precision between these algorithms varies with
increasing bitwidth by enumerating the trends from bitwidth
$n=5$ to $n=10$. We observe that (1) the fraction of inputs
where $R_1 = R_2$ decreases, and (2) \ourmul produces more
precise results than \kernmul for a higher fraction of those
inputs where the outputs differ ($R_1 \neq R_2$) but are
comparable ($R_1 \leqabst R_2 \vee R_2 \leqabst R_1$). The
full results are in the
\Supplement{}~\cite{vishwanathan-tnums-arxiv}. 

Due to the non-associativity of tnum addition
(\Sec{automated-verification}), some orders of adding tnums
are more precise than others, while increasing the number of
tnums added makes the result less precise. Hence, the order
and number of tnums added is significant to the precision of
each multiplication algorithm.  In general, \ourmul is more
precise than both \kernmul and \reghermul due to its
value-mask decomposition. Both \kernmul and \reghermul add
tnums each of which contains both certain and uncertain
trits. Due to the value-mask decomposition, \ourmul
postpones such an addition until the very last step of the
algorithm. Further, \ourmul is more precise than \kernmul
with increasing bitwidth ($n$), since \ourmul has fewer
additions ($n+1$) than \kernmul ($2n$).

\begin{figure}
\centering
\includegraphics[width=0.8\columnwidth]{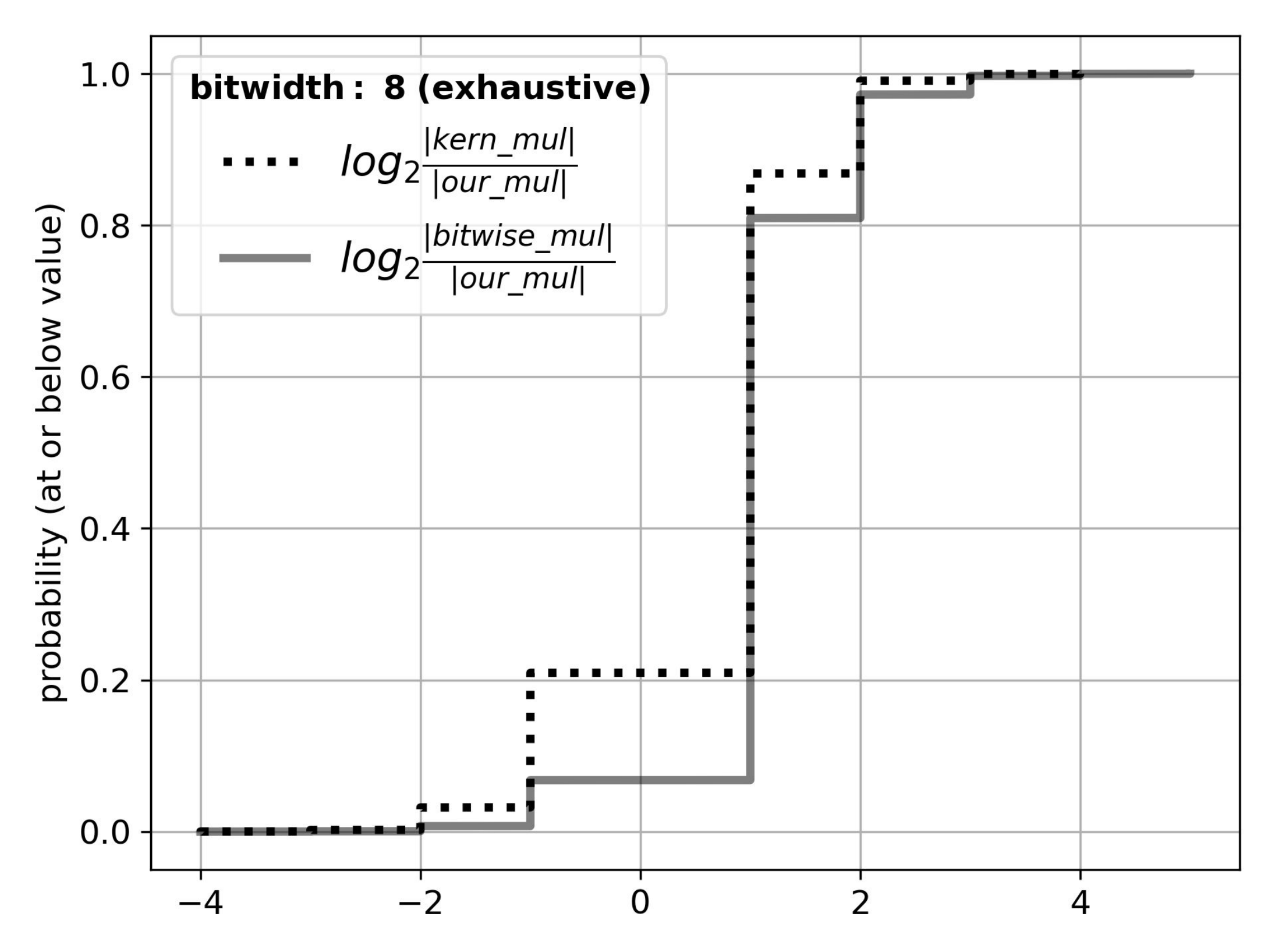}
\caption{\footnotesize Cumulative distribution of the ratio of set
sizes of the tnums (after concretization) produced by (a)
\code{kern_mul} to \code{our_mul}, and (b) \reghermul to
\code{our_mul}, for cases where the output from
\code{our_mul} is different. Results  presented in log scale
($log_{2}$). The input tnums pairs are drawn from the set of
all possible tnum pairs of bitwidth \textbf{8}.}
\label{fig-rel-pres}
\end{figure}

\subsection{Performance evaluation of \ourmul}
\label{sec:eval_performance}

\begin{figure}
\centering
\includegraphics[width=0.8\columnwidth]{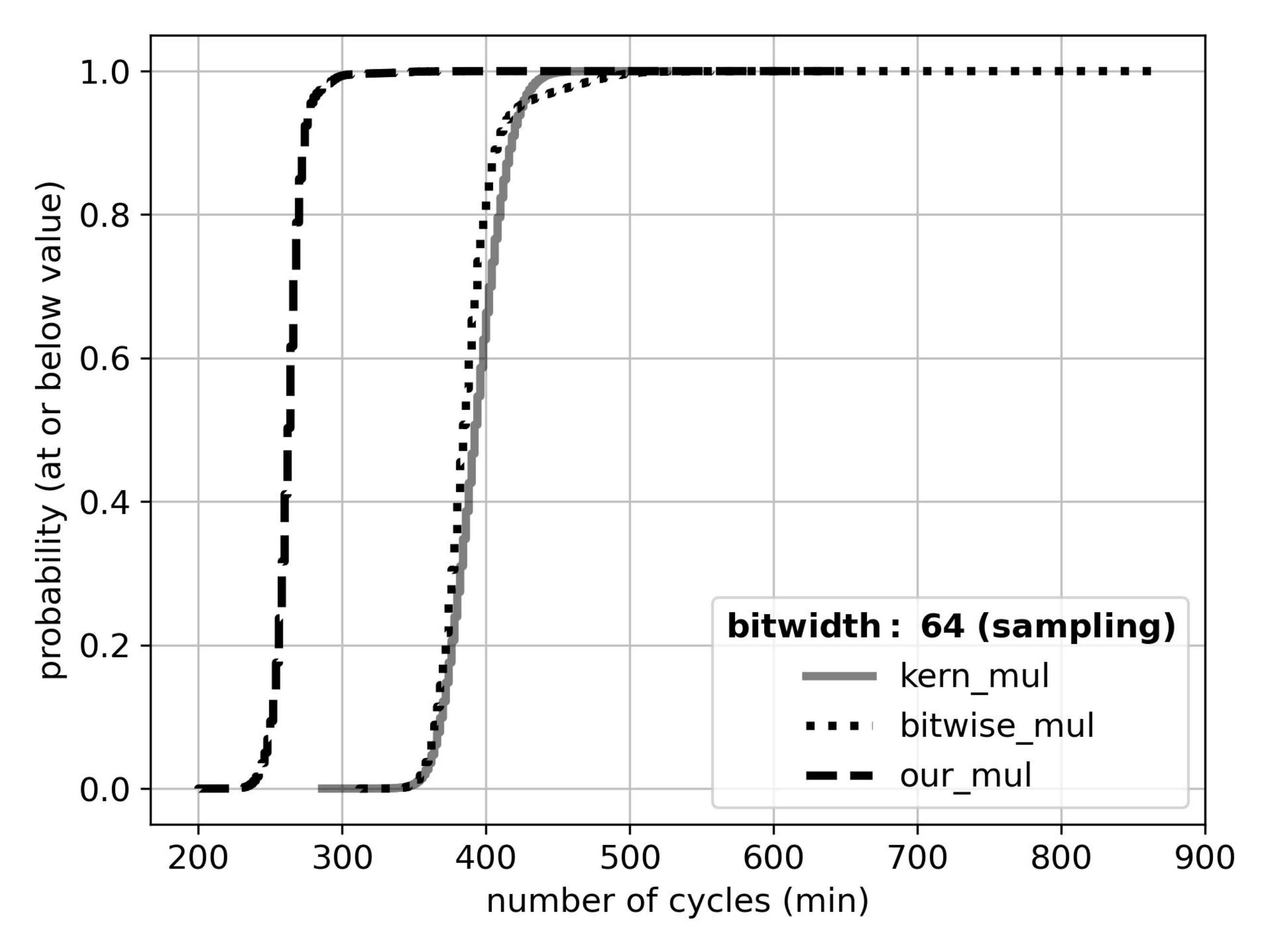}
\caption{\footnotesize Cumulative distribution of the minimum number of CPU cycles
taken by \reghermul, \kernmul, and \ourmul  for 40
million randomly sampled 64-bit input tnum pairs.}
\label{fig-speed}
\end{figure}

We compare the performance (in CPU cycles measured using the
\code{RDTSC} time stamp counter) of all the tnum
multiplication algorithms discussed in this paper:
\kernmul~\cite{tnum-kernel-source},
\reghermul~\cite{regehr-deriving-abstract-transfer}, and our
new algorithm \ourmul.
We perform the experiment using 40 million randomly sampled tnum
pairs (64-bit), repeating each input pair 10 times and measuring the
minimum number of cycles across these trials. Figure~\ref{fig-speed}
reports the cumulative distribution of this cycle count across all the
sampled inputs. All multiplication algorithms have a loop,
and for some algorithms, the number of iterations of the loop depend on
number of unknown bits in the input operands. Hence, the number of
cycles varies across inputs.

We observe that \ourmul is faster (\ie{} fewer CPU cycles
taken) than all the other versions of tnum
multiplication. On average, \kernmul takes around 393
cycles, our optimized version of \reghermul takes 387
cycles, and \ourmul takes 262 cycles (when we take the
average of the minimum of 10 trials for each input tnum
pair).  Our optimizations to \reghermul to use machine
arithmetic were important as it improved the performance
significantly (\ie from 4921 cycles to 387 cycles).  In
summary, efficient use of machine arithmetic and the novel
computation and summation of partial products makes \ourmul
\ourmulfasterthankernmulby
(resp. \ourmulfasterthanreghermulby) faster on average than
\kernmul (resp.  our optimized version of \reghermul).

\section{Related Work}

\textbf{BPF safety}. Given the widespread use of BPF, recent
efforts have explored building safe JIT compilers and
interpreters~\cite{jitterbug-osdi20,jitsynth-cav20,serval-sosp19,jitk-osdi14}.
These works assume the correctness of the in-kernel static
checker and the JIT translation happens after the BPF code
passes the static checker.
Prevail~\cite{untrusted-extensions-pldi19} proposes an
offline BPF analyzer  using abstract interpretation with the
zone abstract domain and supports programs with loops. In
contrast to this paper, these prior efforts have not looked
at verifying the tnum operations in the Linux kernel's
static analyzer or explored the tnum domain specifically.

\textbf{Abstract interpretation.} Many static
analyses use an abstract domain for abstract
interpretation~\cite{cousot-1977,cousot-lecture-notes,cousot-2001}.
Abstract domains like intervals, octagons~\cite{mine-tutorial}, and
polyhedra-based domains~\cite{polyhedra-popl17} enhance the precision
and efficiency of the underlying operations.  Unlike the Linux
kernel's tnums, their intended use is in offline settings. Of
particular relevance to our work is the known-bits domain from LLVM
\cite{llvm-known-bits-analysis, testing-static-analysis-cgo20,
  dataflow-pruning-oopsla20}, which, like tnums, is used to reason
about the certainty of individual bits in a value. Our work on
verifying tnums will be likely useful to LLVM's known-bits analysis,
as prior work does not provide proofs of precision and
soundness for
arithmetic operations such as addition and multiplication.

\textbf{Safety of static analyzers.} 
One way to check for soundness and precision bugs in static
analyzers is to use automated random
testing~\cite{klinger-2019,cuoq-2012}. Recently, Taneja et
al.~\cite{testing-static-analysis-cgo20} test dataflow analyses in
LLVM to find soundness and precision bugs by using an
SMT-based test oracle~\cite{souper}. Bugariu et 
al. \cite{testing-numerical-abstract-domains} test the soundness and
precision of widely-used numerical abstract domains
\cite{mine-HOSC06,polyhedra-popl17}. They use mathematical properties
of such domains as test oracles while comparing against a set of
representative domain elements. They assume that the oracle
specification of operations is correct and precise. This paper differs
from these approaches in that we formalize and construct analytical
proofs for the abstract operations.

\section{Conclusion}
Abstract domains like tnums are widely used to track register values
in the Linux kernel and in various compilers. This paper performs
verification of tnum arithmetic operations, and develops a new
implementation for tnum multiplication. Our algorithm for tnum
multiplication is sound, precise, and faster than Linux's kernel
multiplication. Our new multiplication algorithm is now part of the
Linux kernel.

\section*{Acknowledgments}
This paper is based upon work supported in part by the
National Science Foundation under FMITF-Track I Grant
No.~2019302 and the Facebook's Networking Systems Research
Award. We thank Edward Cree for his feedback on strength
reductions and precision improvements for tnum operations.

\bibliographystyle{IEEEtran}
\bibliography{refs}
\section{Supplementary Material}
\label{sec:supplement}

\subsection{Proofs for Auxiliary Lemmas for Tnum Addition}

\label{appendix:tnum-add}
We provide proofs for the additional lemma required to prove
the soundness and maximal precision of tnum addition. Let
$x[i]$ denote the \ith{i} bit of an integer $x \in
\bb{Z}{n}$. 

Note that by our definition of $\alpha$ and $\gamma$, given
a non-empty tnum $T \in \bb{Z}{n} \times \bb{Z}{n}$, the
following holds. 

\begin{align}
\begin{split}
\label{eqn:tnum_value_corresp}
T.v[i] = 1 \; \Leftrightarrow \;
& \forall c \ \in \gamma(T): c[i] = 1 \\
T.v[i] = 0 \wedge T.m[i] = 0 \; \Leftrightarrow \;
& \forall c \in \gamma(T): c[i] =  0 \\
T.v[i] = 0 \wedge T.m[i] = 1 \; \Leftrightarrow \;
& \forall c \in \gamma(T):  c[i] =  0 \vee c[i] =  1 \\
\end{split}
\end{align}

Suppose $p$ and $q$ are two concrete values in tnum $P$ and
tnum $Q$, respectively, i.e., $p \in \gamma(P ), q \in
\gamma(Q)$. 

\begin{lemma}
\label{lemma:minimum-carries-lemma-appendix}
\emph{(\Lemma{lemma:minimum-carries-lemma} in the main
text).} 
\textbf{Minimum carries lemma.} 
The addition $sv = P.v + Q.v$ will produce a
sequence of carry bits that has the least number of 1s out
of all possible additions $p + q$.
\end{lemma}

\begin{proof}
We prove any concrete addition $p + q$ will produce a
sequence of carry bits with 1s in at least those positions
where $sv$ produced carry bits set to 1.

Let $\ithbit{c}{{sv}}{i}$ denote the carry out bit at the
\ith{i} position produced by $P.v + Q.v$ and
$\ithbit{c}{{out}}{i}$  denote the carry out bit at the
\ith{i} position produced by $p + q$. We will prove by
induction that if $P.v[i] + Q.v[i]$ produces a carry bit
then so must $p[i] + q[i]$. 

\textbf{Base case}: At bit position $i=0$, $P.v[i] + Q.v[i]$
produces a carry bit only when $P.v[i]=1$ and $Q.v[i]=1$.
Since $P.v[i]$ is $1$, and $Q.v[i]$ is $1$, by
\Eqn{eqn:tnum_value_corresp}, this implies that $\forall p:
p[i] = 1$ and that $\forall q: q[i] = 1$. Hence, $p[i] +
q[i]$ must also produce a carry bit. Hence,
$\ithbit{c}{{sv}}{0} = 1 \rightarrow \ithbit{c}{{out}}{0} =
1$.

\textbf{Induction step}: Assume that for all $0 \leq j \leq
i-1$, $\ithbit{c}{{sv}}{j} = 1 \rightarrow
\ithbit{c}{{out}}{j} = 1$. Now, we show that
$\ithbit{c}{{sv}}{i} = 1 \rightarrow \ithbit{c}{{out}}{i}
= 1$ by considering all possible cases:
  
\begin{enumerate}
\item $P.v[i] = 0$ and $Q.v[i] = 0$. Irrespective of any
carry-in $c_{in}[i]$, $0 + 0 + c_{in}[i]$ can never produce
a carry-out. $\ithbit{c}{{sv}}{i} = 1 \rightarrow
\ithbit{c}{{out}}{i} = 1$ holds vacuously. 

\item $P.v[i] = 1$ and $Q.v[i] = 1$. Irrespective of any
carry-in $c_{in}[i]$, $1 + 1 + c_{in}[i]$ will always
produce a carry-out. Now, since is $P.v[i]$ is $1$, and
$Q.v[i]$ is $0$, by \Eqn{eqn:tnum_value_corresp}, implies
that $\forall p: p[i] = 1$ and $\forall q: q[i] = 1$. Hence
all concrete additions $p[i] + q[i]$ must produce a
carry-out. Hence, $\ithbit{c}{{sv}}{i} = 1 \rightarrow
\ithbit{c}{{out}}{i} = 1$ holds. 

\item $(P.v[i] = 1 \wedge Q.v[i] = 0) \vee (P.v[i] = 0
\wedge Q.v[i] = 1)$. From \Eqn{eqn:tnum_value_corresp}, this
implies that $\forall p, q. \, (p[i] = 1 \wedge (q[i] = 1
\vee q[i] = 0)) \vee ((p[i] = 1 \vee q[i] = 0) \wedge q[i] =
1)$. That is, either $p[i] = 1$ and $q[i]$ is unknown, or
$q[i] = 1$ and $p[i]$ is unknown. This case entails two
possibilites:

\begin{enumerate}
\item $\ithbit{c}{{sv}}{i-1} = 0$. This implies that
$\ithbit{c}{{sv}}{i} = 0$, since there is no carry-out
produced by the additions $1 + 0 + 0$ or $0 + 1 + 0$.
$\ithbit{c}{{sv}}{i} = 1 \rightarrow \ithbit{c}{{out}}{i}
= 1$ holds vacuously. 
  
\item $\ithbit{c}{{sv}}{i-1} = 1$. . Now,
$\ithbit{c}{{sv}}{i} = P.v[i] + Q.v[i] + c_{sv}[i-1]$.
Hence, $\ithbit{c}{{sv}}{i}$ is the carry-out produced by $1
+ 0 + 1$ or $0 + 1 + 1$, which always evalues to $1$. Now,
by the induction hypothesis $\ithbit{c}{{out}}{i-1} = 1$. We
know that $\ithbit{c}{{out}}{i} = p[i] + q[i] +
c_{out}[i-1]$. Hence, $\ithbit{c}{{out}}{i}$  is the
carry-out produced by either $1 + q[i] + 1$ or $p[i] + 1 +
1$, which always evaluates to $1$. Hence
$\ithbit{c}{{sv}}{i} = 1 \rightarrow \ithbit{c}{{out}}{i} =
1$ holds.
\end{enumerate}
\end{enumerate}
  
Hence, it follows that if $P.v[i] + Q.v[i]$ produces a carry
bit then so must $p[i] + q[i]$, for all $p \in \gamma(P)$
and  $q \in \gamma(Q)$.
\end{proof}

\begin{lemma}
\label{lemma:maximum-carries-lemma-appendix}
\emph{(\Lemma{lemma:maximum-carries-lemma} in the main
text).} 
\textbf{Maximum carries lemma.} 
The addition $\Sigma = (P.v + P.m) + (Q.v + Q.m)$ will
produce the sequence of carry bits with the most number of
1s out of all possible additions $p$ + $q$. 
\end{lemma}
  
\begin{proof}
We prove that any $p$ + $q$ will produce a sequence of carry
bits with 1s in at most those positions where $\Sigma$
produced carry bits set to 1.

Let $\ithbit{c}{{\Sigma}}{i}$ denote the carry out bit at
the \ith{i} position produced by $P.v + Q.v$ and
$\ithbit{c}{{out}}{i}$  denote the carry out bit at the
\ith{i} position produced by $p + q$. We will prove by
induction that if $\ithbit{c}{{\Sigma}}{i}$ does not produce
carry bit then neither will $\ithbit{c}{{out}}{i}$.

\textbf{Base case}: At bit position $i=0$,
$\ithbit{c}{{\Sigma}}{0} = (P.v[0] + P.m[0]) + (Q.v[0] +
Q.m[0])$ does not produces produce a carry bit only when
$(\tvi{P}[0] = 0 \wedge \tmi{P}[0] = 0) \vee (\tvi{Q}[0] =0
\wedge \tmi{Q}[0] = 0)$. This implies that either $p[i] = 0
\vee q[i] = 0$. Hence $p[0] + q[0]$ will also not produce a
carry. Hence $\ithbit{c}{{\Sigma}}{0} = 0 \rightarrow
\ithbit{c}{{out}}{0} = 0$.

\textbf{Induction step}: Assume that for all $0 \leq j \leq
i-1$, $\ithbit{c}{{\Sigma}}{j} = 0 \rightarrow
\ithbit{c}{{out}}{j} = 0$. Now, we show that
$\ithbit{c}{{\Sigma}}{i} = 0 \rightarrow \ithbit{c}{{out}}{i}
= 0$ by considering all possible cases:
    
\begin{enumerate}
\item $P.v[i] + P.m[i]= 0$ and $Q.v[i] + Q.v[i] = 0$.
Irrespective of any carry-in $c_{in}[i]$, $0 + 0 +
c_{in}[i]$ can never produce a carry-out. Now, $P.v[i] +
P.m[i]= 0$ holds only when $P.v[i] = 0 \wedge \tmi{Q}[i] =
0$. By \Eqn{eqn:tnum_value_corresp} $p[i] = 0 \wedge q[i] =
0$. Hence, irrespective of any carry-in $p[i] + q[i]$ will
not produce a carry-out. Hence, $\ithbit{c}{{\Sigma}}{i} = 0
\rightarrow \ithbit{c}{{out}}{i} = 0$ holds. 
  
\item $P.v[i] + P.m[i]= 1$ and $Q.v[i] + Q.v[i] = 1$.
Irrespective of any carry-in $c_{in}[i]$, this will always
produce a carry-out. can never produce a carry-out. Hence,
$\ithbit{c}{{\Sigma}}{i} = 0 \rightarrow \ithbit{c}{{out}}{i} =
0$ holds vacuously.

\item $P.v[i] + P.m[i]= 0$ and $Q.v[i] + Q.v[i] = 1$.
This case entails two possibilites:
\begin{enumerate}
\item $\ithbit{c}{{\Sigma}}{i-1} = 1$. This implies that
$\ithbit{c}{{\Sigma}}{i} = 1$. Hence $\ithbit{c}{{\Sigma}}{i} =
0 \rightarrow \ithbit{c}{{out}}{i} = 0$ holds vacuously.
\item $\ithbit{c}{{\Sigma}}{i-1} = 0$.  This implies that
$\ithbit{c}{{\Sigma}}{i} = 0$. Note because $P.v[i] +
P.m[i]= 0$ we have $P.v[i] = 0$ and $P.m[i]= 0$, and by
\Eqn{eqn:tnum_value_corresp}, this implies  $p[i] = 0$. Now,
by the induction hypothesis $\ithbit{c}{{out}}{i-1} = 0$.
Hence, irrespective of the value of $q[i]$,
$\ithbit{c}{{out}}{i}$ does not produce a carry. Hence
$\ithbit{c}{{\Sigma}}{i} = 0 \rightarrow
\ithbit{c}{{out}}{i} = 0$ holds. 
\end{enumerate}

\item $P.v[i] + P.m[i]= 1$ and $Q.v[i] + Q.v[i] = 0$.
This case is completely symmetrical to case 3. 
\end{enumerate}

Hence, it follows that if $(P.v + P.m) + (Q.v + Q.m)$ does
not produce a carry, then neither will $p[i] + q[i]$, for
all $p \in \gamma(P)$ and  $q \in \gamma(Q)$.
\end{proof}

\begin{lemma}
\label{lemma:capture-uncertainty-lemma-appendix}
\emph{(\Lemma{lemma:capture-uncertainty-lemma} in the main
text).} 
\textbf{Capture uncertainty lemma.} 
Let $sv_{c}$ and $\Sigma_{c}$ be the sequence of carry-in
bits from the additions in $sv$ and $\Sigma$, respectively.
Suppose $\chi_c \triangleq sv_{c} \,\bitwiseXor\,
\Sigma_{c}$. The bit positions $k$ where
$\ithbitx{\chi_c}{k} = 0$ have carry bits fixed in all
concrete additions $p+q$ from $+(\gamma(P), \gamma(Q))$. The
bit positions $k$ where $\ithbitx{\chi_c}{k} = 1$ vary
depending on the concrete addition: \ie $\exists p_1, p_2
\in \gamma(P), q_1, q_2 \in \gamma(Q)$ such that $p_1 + q_1$
has its carry bit set at position $k$ but $p_2 + q_2$ has
that bit unset.
\end{lemma}

\begin{proof}
Consider the carry sequences $sv_{c}$ and $\Sigma_{c}$. At
any given bit position $i$, we consider the following cases:

\begin{enumerate}
\item $\Sigma_{c}[i] = 0$, and $sv_{c}[i] = 0$. From the
maximum carries lemma (Lemma
\ref{lemma:maximum-carries-lemma-appendix}), $\Sigma_{c}[i]
= 0$ implies that \emph{no} addition $p[i-1] + q[i-1] +
c_{in}[i-1]$ (where $\Tnumin{p}{P}$ and $\Tnumin{q}{Q}$, and
$c_{in}[i]$ is the carry-out from bit position $i-2$)
produced a carry-out from position $i-1$. Equivalently, the
carry-in at bit position $i$ is \emph{always} unset. That
is, \emph{all} concrete additions $p + q$ have a fixed
carry-in bit of $0$ at position $i$. 
\item $\Sigma_{c}[i] = 1$, and $sv_{c}[i] = 1$. From the
minimum carries lemma, (Lemma
\ref{lemma:minimum-carries-lemma-appendix}), $sv_{c}[i] = 1$
implies that \emph{all} additions $p[i-1] + q[i-1] +
c_{in}[i-1]$ (where $\Tnumin{p}{P}$ and $\Tnumin{q}{Q}$, and
$c_{in}[i-1]$ is the carry-out from bit position $i-2$)
produced a carry-out from position $i-1$. Equivalently, the
carry-in at bit position $i$ is \emph{always} set. That is,
\emph{all} concrete additions $p + q$ have a fixed carry-in
bit of $1$ at position $i$. 
\item $\Sigma_{c}[i] = 1$, and $sv_{c}[i] = 0$. 
Note that it always holds that $\Tnumin{P.v}{P}$ and
$\Tnumin{Q.v}{Q}$. Since $sv_{c}[i] = 0$, we know one
concrete addition  $p + q$ that definitely does not produce
a carry-out from bit position $i-1$: setting $p = P.v$ and
$q = Q.v$. Now note also that it always holds that
$\Tnumin{P.v + P.m}{P}$ and $\Tnumin{Q.v + Q.m}{Q}$. Since
$\Sigma_{c}[i] = 1$, we know one concrete addition  $p + q$
that definitely produces a carry-out from bit position
$i-1$: setting $p = P.v + P.m$ and $q = Q.v + Q.m$. 
\item $\Sigma_{c}[i] = 0$, and $sv_{c}[i] = 1$. This case is
never possible, because if the additions in $sv[i-1]$
produce a carry out (it is given that $sv_{c}[i] = 1$), then
so must the additions in $\Sigma[i-1]$ (but it is given that
$\Sigma_{c}[i] = 0$). 
\end{enumerate}

From cases 1) and 2) above, it is clear that if
$\Sigma_{c}[i]$ and $sv_{c}[i]$ are the same \ie if
$\chi_{c}[i] = 0$, then all the additions at position $i-1$
exclusively (i) produce a carry-out, or (ii) do not produce
a carry-out. From case 3) it is clear that if
$\Sigma_{c}[i]$ and $sv_{c}[i]$ differ \ie if $\chi_{c}[i] =
1$, then some additions at position $i-1$ produce a
carry-out and some additions do not produce a carry-out.
This proves our result.
\end{proof}

For the next lemma, recall the definitions of the full 
adder equations. 

\begin{definition}
\label{addition-of-bits-appendix}
\emph{(Definition~\ref{addition-of-bits} in the main text).}
\textbf{Full adder equations.}
When adding two concrete binary numbers $p$ and $q$, each
bit of the addition result $r$ is set according to the
following:
$$r[i] = p[i] \,\bitwiseXor\, q[i] \,\bitwiseXor\,
\ithbit{c}{{in}}{i}$$ where $\bitwiseXor$ is the
exclusive-or operation and $\ithbit{c}{{in}}{i} =
\ithbit{c}{{out}}{i-1}$ and $\ithbit{c}{{out}}{i-1}$ is the
carry-out from the addition in bit position $i-1$. The
carry-out bit at the $i^{th}$ position is given by
$$\ithbit{c}{{out}}{i} = (p[i] \,\bitwiseAnd\, q[i])
\,\bitwiseOr\, (\ithbit{c}{{in}}{i} \,\bitwiseAnd\, (p[i]
\,\bitwiseXor\, q[i]))$$
\end{definition}

\begin{lemma}
\label{thm:add-soundness-maximal-precision-appendix}
\emph{(\Lemma{thm:add-soundness-maximal-precision} in the
main text).} 
\textbf{Equivalence of mask expressions.} 
Tnum addition (\code{tnum_add}) uses $(sv \,\bitwiseXor\,
\Sigma) \,\bitwiseOr\, P.m \,\bitwiseOr\, Q.m$ to compute
the mask of the result. It is in fact the case that $(sv
\,\bitwiseXor\, \Sigma) \,\bitwiseOr\, P.m \,\bitwiseOr\,
Q.m$ and $(sv_{c} \,\bitwiseXor\, \Sigma_{c}) \,\bitwiseOr\,
P.m \,\bitwiseOr\, Q.m$ compute the same result.
\end{lemma}

\begin{proof}
We have to prove that given $sv_{c}[i+1] \, \bitwiseXor
\Sigma_{c}[i+1] = 1$, this implies that there exist $p_{1},
p_{2}, q_{1}, q_{2}$ such that $p_{1}[i] + q_{1}[i] +
c_{in}[i]$ produces a carry out, and  $p_{2}[i] + q_{2}[i] +
c_{in}[i]$ does not produce a carry out.

Recall that $sv = P.v + Q.v$ and $\Sigma = (P.v + P.m) +
(Q.v + Q.m)$ Let $sv_{c}$ and $\Sigma_{c}$ be the sequence
of carry-in bits from the additions in $sv$ and $\Sigma$.
respectively. We prove that the terms $(sv_c\ \bitwiseXor\
\Sigma_c) \,\bitwiseOr\, P.m \,\bitwiseOr\, Q.m$ and $(sv\
\bitwiseXor\ \Sigma) \,\bitwiseOr\, P.m \,\bitwiseOr\, Q.m$
are exactly the same.

From Definition~ \ref{addition-of-bits-appendix} of the
full-adder equation, we have the addition result $r[i] =
p[i] \,\bitwiseXor\, q[i] \,\bitwiseXor\,
\ithbit{c}{{in}}{i}$. Hence, the \ith{i} bit of $sv$,
$sv[i]$ is $P.v[i] \,\bitwiseXor\, Q.v[i] \,\bitwiseXor\,
sv_c[i]$. Similarly, \ith{i} bit of $\Sigma[i]$ can be
written as $(P.v[i] + P.m[i]) \,\bitwiseXor\, (Q.v[i] +
Q.m[i]) \,\bitwiseXor\, \Sigma_c[i]$. Consider $sv[i]
\,\bitwiseXor\, \Sigma[i]$.

\begin{small}
\begin{equation*}
\begin{aligned}
sv[i] \,\bitwiseXor\, \Sigma[i] = & P.v[i] \,\bitwiseXor\, Q.v[i] \,\bitwiseXor\, 
sv_c[i] \\ & \,\bitwiseXor\, (P.v[i] + P.m[i]) 
\,\bitwiseXor\, (Q.v[i] + Q.m[i]) \,\bitwiseXor\,
\Sigma_c[i]
\end{aligned}
\end{equation*}
\end{small}

Now $sv[i] \,\bitwiseXor\, \Sigma[i]$ can differ from
$sv_{c}[i] \,\bitwiseXor\, \Sigma_{c}[i]$, if their
exclusive-or is $1$, \ie $(sv[i] \,\bitwiseXor\, \Sigma[i])
\,\bitwiseXor\, (sv_{c}[i] \,\bitwiseXor\, \Sigma_{c}[i])$
is $1$.

\begin{small}
\begin{equation*}
\begin{aligned}
& (sv[i] \,\bitwiseXor\, \Sigma[i]) \,\bitwiseXor\, 
  (sv_c[i] \,\bitwiseXor\, \Sigma_c[i]) \\
& = P.v[i] \,\bitwiseXor\, Q.v[i] \,\bitwiseXor\, sv_c[i] \\ 
& \; \; \;  \; \; \; \,\bitwiseXor\, (P.v[i] + P.m[i]) 
 \,\bitwiseXor\, (Q.v[i] + Q.m[i]) \,\bitwiseXor\, 
\Sigma_c[i] \\
& \; \; \;  \; \; \; \,\bitwiseXor\, sv_c[i] \,\bitwiseXor\, \Sigma_c[i] \\
& = P.v[i] \,\bitwiseXor\, Q.v[i] \,\bitwiseXor\, 
  (P.v[i] + P.m[i]) \,\bitwiseXor\, (Q.v[i] + Q.m[i]) \\
& = P.v[i] \,\bitwiseXor\, (P.v[i] + P.m[i]) 
  \,\bitwiseXor\, Q.v[i] \,\bitwiseXor\, (Q.v[i] + Q.m[i]) 
\end{aligned}
\end{equation*}
\end{small}

Now consider the sub-term in the above formula $P.v[i]
\,\bitwiseXor\, (P.v[i] + P.m[i])$. For any non-empty tnum
$P.v[i] \, \bitwiseAnd \, P.m[i] = 0$. Hence the term
$P.v[i] \,\bitwiseXor\, (P.v[i] + P.m[i])$ reduces to
$P.m[i]$. Similarly, term $Q.v[i] \,\bitwiseXor\, (Q.v[i] +
Q.m[i])$ reduces to $Q.m[i]$. 

Thus, $sv[i] \,\bitwiseXor\, \Sigma[i]$ differs from
$sv_c[i] \,\bitwiseXor\, \Sigma_c[i]$ when $P.m[i]
\,\bitwiseXor\, Q.m[i] = 1$. Now, if $P.m[i] \,\bitwiseXor\,
Q.m[i] = 1$ then $P.m[i] \,\bitwiseOr \,Q.m[i] = 1$ as well.
Thus, setting $x \triangleq sv \,\bitwiseXor\, \Sigma$ or $x
\triangleq sv_{c} \,\bitwiseXor\, \Sigma_{c}$ is immaterial
to the value $x \,\bitwiseOr\, P.m \,\bitwiseOr\, Q.m.$.
\end{proof}

\begin{lemma}
\label{thm:soundness-optimality-of-tnum-add-appendix}
\emph{(\Lemma{thm:soundness-optimality-of-tnum-add} in the
main text).} 
\textbf{Soundness and optimality of tnum\_add.} 
The algorithm \code{tnum_add} shown in \Lst{tnum_add} is a
sound and optimal abstraction of concrete addition over
\nbit bitvectors.
\end{lemma}

\begin{proof}
Since \code{tnum_add} captures all, and only, the
uncertainty in the concrete results of tnum addition, it
is sound and optimal.    
\end{proof}

\subsection{Proof of our new algorithm for Tnum multiplication}
\label{sec:tnum-mul-full-proof-appendix}

In this section, we present the complete proof of our tnum
multiplication algorithm in \Lst{our_mul_simplified}. Recall
first, the definition of our concretization function
$\gamma$. 

\begin{align}
\begin{split}
\label{eqn_gamma_appendix}
\gamma(t) \; = \; \gamma(\tnumi{t}) \; & \triangleq \; \big\{ c \in \mathbb{Z} \mid c \; \bitwiseAnd \; \bitwiseNot \tmi{t} = \tvi{t} \big\} \\
\end{split}
\end{align}

\begin{observation}
\label{obs:partial-products-appendix} 
\emph{(\Obs{partial-products} in the main text).}
For two concrete bitvectors $x$ and $y$ of width $n$ bits,
the result of multiplication $y * x$ is just
$$ y * x = \sum_{k=0}^{n-1} \ithbitx{x}{k} * (y
\,\bitwiseLshift\, k)$$
\end{observation}

We call each term in the summation $\ithbitx{x}{k} * (y \ll
k)$ a {\em partial product.}

\begin{lemma}
\label{lemma:union-tnum-with-zero-appendix}
\emph{(\Lemma{lemma:union-tnum-with-zero} in the main
text).} \textbf{Tnum set union with zero.} Given a non-empty
tnum $P \in \bb{Z}{n} \times \bb{Z}{n}$, define $Q
\triangleq tnum(0, P.v \mid P.m)$. Then, (i) $x \in
\gamma(P) \Rightarrow x \in \gamma(Q)$, and (ii) $0 \in
\gamma(Q)$.
\end{lemma}

\begin{proof} By the definition of $\gamma$ in
\Eqn{eqn_gamma_appendix}, for a tnum $T$ such that $\tvi{T}
= 0$, $\gamma(0, \tmi{T}) = \{ c \mid c \bitwiseAnd
\bitwiseNot \tmi{T}  = 0 \}$. Since $c = 0$ satisfies the
condition, it is true that $0 \in \gamma(Q)$. Additionally,
we can see intuitively that building a new tnum $Q$ whose
mask has set all the bits corresponding to the set value or
mask bits of $P$ ensures that $\gamma(P) \subseteq
\gamma(Q)$. 

More formally, consider any concrete value $\Tnumin{x}{P}$.
we observe that it is sufficient to prove the membership of
$x$ in tnum $\gamma(Q)$ bit-by-bit, since $\gamma$ is
bitwise exact. Note that $Q.v = 0$. For each bit position
$k$, there are three possible cases.

\begin{enumerate}
\item $\tvi{P}[k] = 1, \tmi{P}[k] = 0$. This implies that
$x[k] = 1$, due to the bitwise exactness of $\gamma$. By the
way of constructing $Q$, we have $\tvi{Q}[k] = 0$, and
$\tmi{Q}[k] = 1$. $x[k] = 1$ satisfies $\gamma(\tvi{Q}[k],
\tmi{Q}[k])$. Hence $x[k] \in \gamma(Q[k])$. 
\item $\tvi{P}[k] = 0, \tmi{P}[k] = 0$. This implies that
$x[k] = 0$, due to the bitwise exactness of $\gamma$. By the
way of constructing $Q$, we have $\tvi{Q}[k] = 0$, and
$\tmi{Q}[k] = 0$. $x[k] = 0$ satisfies $\gamma(\tvi{Q}[k],
\tmi{Q}[k])$. Hence $x[k] \in \gamma(Q[k])$. 
\item $\tvi{P}[k] = 0, \tmi{P}[k] = 1$.  By the way of
constructing $Q$, we have $\tvi{Q}[k] = 0$, and $\tmi{Q}[k]
= 1$. Whatever the value of $x[k]$, it satisfies
$\gamma(\tvi{Q}[k], \tmi{Q}[k])$. Hence $x[k] \in
\gamma(Q[k])$. 
\end{enumerate}

Hence, we can say that without loss of generality, for all
bit positions $k$, $x[k] \in \gamma(P[k]) \Rightarrow x[k]
\in \gamma(Q[k])$. Hence $x \in \gamma(P) \Rightarrow x \in
\gamma(Q)$.

\end{proof}

\begin{lemma}
\label{lemma:value-mask-decoupled-additions-appendix}
\emph{(\Lemma{lemma:value-mask-decoupled-additions} in
the main text).} \textbf{Value-mask-decomposed tnum
summations.} Given $n$ non-empty tnums $T_0, T_1, \ldots,
T_{n-1} \in \bb{T}{n}$. Suppose we pick $n$ values $z_0,
z_1, z_2, \ldots, z_{n-1} \in \bb{Z}{n}$ such that $\forall
\, 0 \leq j \leq n-1: \, z_j \in \gamma(T_j)$. Define tnum
\begin{align}
\begin{split}
S \;\; \triangleq \;\;
\code{tnum_add}(\code{tnum_add}_{j=0}^{n-1} \;\;(tnum(T_j.v,
0)),
\\ \;\; \;\; \;\;\code{tnum_add}_{j=0}^{n-1} \;\;(tnum(0,
T_j.m)))
\end{split}
\end{align}
where $\code{tnum_add}_{(\cdot)}^{(\cdot)}$ is a
variable-arity version of \code{tnum_add} defined
above. Then, $\sum_{j=0}^{n-1} z_j \in \gamma(S)$.
\end{lemma}
    
\begin{proof}
Intuitively, suppose we had $n$ tnums $T_i, 0 \leq i \leq
n-1$ and we seek to construct a new tnum $S$ whose
concretization $\gamma(S)$ contains all possible {\em
concrete sums} from the $T_i$, \ie such that
$\{\sum_{j=0}^{n-1} x_j \mid x_i \in \gamma(T_i)\} \leqconc
\gamma(S)$. The most natural method to construct such a tnum
$S$ is to use the sound abstract addition operator
\code{tnum_add} over the $T_i$, \ie
$\code{tnum_add}_{j=0}^{n-1}(T_j)$. This lemma provides
another method of constructing such a tnum $S$: decompose
the tnums $T_i$ each into two tnums, consisting of the
values and the masks separately. Use \code{tnum_add} to
separately add the value tnums, add the mask tnums, and
finally add the two resulting tnums from the value-sum and
mask-sum to produce $S$. 

Note that given a non-empty tnum $T$, its concretization
$\gamma(T)$ contains $T.v$ and also, $T.v$ is the smallest
element in $\gamma(T)$.

\begin{align}
\label{eqn:tv_is_smallest}
\Tnumin{x}{T.v} \wedge \forall \Tnumin{x}{T}, x \geq \tvi{T}
\end{align}

For proving
\Lemma{lemma:value-mask-decoupled-additions-appendix}, we
start by proving the simpler property below.

\proofheading{Property P0 (value-mask decomposition of a
single tnum).} Given a non-empty tnum $T$. For any
$\Tnumin{x}{T}$, we can decompose the concrete value $x$ as
the sum of two concrete values, \ie $x = x' + x''$ where $
\Tnumin{x'}{T.v, 0}$ and $\Tnumin{x''}{0, T.m}$ .

\begin{proof}
For any $\Tnumin{x}{T}$, from \Eqn{eqn:tv_is_smallest} it
must be that $x \geq T.v$. Additionally, consider. $(x -
\tvi{T})$. The only bits that may be set in $(x - \tvi{T})$
are those that correspond to unknown trits in  $T$, which
implies that the $\tmi{T} = 1$, \ie $\forall k
\in \{ 0, 1 \}, (x - \tvi{T}) = 1 \rightarrow T[i] = \mu
\rightarrow \tmi{T}[i] = 1$. Hence, it is true that $(x -
\tvi{T}) \bitwiseAnd \bitwiseNot \tmi{T} = 0$. 

In particular, setting $y = x - T.v$, and rewriting the
equation above, $(y - 0) \bitwiseAnd \bitwiseNot \tmi{T} =
\tvi{T}$. By definition of $\gamma$
(\Eqn{eqn_gamma_appendix}), this implies that
$\Tnumin{y}{tnum(0, T.m)} $. Since $x = y + T.v$, we can set
$x' = T.v$ and $x'' = y$ to prove the result.
\end{proof}

Using the value-mask decomposition property above, we can
rewrite all the concrete $z_j = z_j' + z_j''$, such that
$\Tnumin{z_j'}{tnum(T_j.v, 0)}$ and $\Tnumin{z_j''}{tnum(0,
T_j.m)}$. Assuming that tnum addition is sound, it is the
case that

\begin{enumerate}
\setlength\itemsep{1em}
\item $\Tnumin{\sum_{j=0}^{n-1}
z_j'}{\code{tnum_add}_{j=0}^{n-1} (tnum(T_j.v,0)}$,
\item $\Tnumin{\sum_{j=0}^{n-1}
z_j''}{\code{tnum_add}_{j=0}^{n-1}(tnum(0, T_j.m))}$, 

\item Hence, it follows that
\begin{align*}
{\sum_{j=0}^{n-1} z_j' + \sum_{j=0}^{n-1} z_j''} \in \\
\gamma \Big(\code{tnum_add} \big(
\code{tnum_add}_{j=0}^{n-1} (tnum(T_j.v,0)), \\
\code{tnum_add}_{j=0}^{n-1} (tnum(0, T_j.m))
\big)
\Big)
\end{align*}
\end{enumerate}

Moreover, by the associativity of arithmetic addition,
$\sum_{j=0}^{n-1} z_j' + \sum_{j=0}^{n-1} z_j''$ is just
$\sum_{j=0}^{n-1} z_j$, since each $z_j = z_j' + z_j''$.
This proves our result. 

\end{proof}

\begin{theorem}
\label{thm:our-mul-soundness-appendix}
\emph{(\OurTheorem{thm:our-mul-soundness} in the main
text).} \textbf{Soundness of \ourmul.} $\forall \, x \in
\gamma(P)$, $y \in \gamma(Q)$ the result $R$ returned by
\ourmulsimplified (\Lst{our_mul_simplified}) is such that $x
* y \in \gamma(R)$, assuming that abstract tnum addition
(\code{tnum_add}) and abstract tnum shifts
(\code{tnum_lshift}, \code{tnum_rshift}) are sound.
\end{theorem}

\begin{proof}
At a high level, the algorithm constructs two tnums at each
iteration of the loop. Intuitively, the first tnum $ACC_V$
sums the concrete bits of each partial product, while the
second tnum $ACC_M$ sums the uncertain bits of each partial
product using tnum addition. Then, $\code{tnum_add}(ACC_V,
ACC_M)$ is a tnum which contains all possible sums of the
partial products, \ie all possible products $x * y$ where
$\Tnumin{x}{P}$ and $\Tnumin{y}{Q}$.

More formally, our proof proceeds by establishing three
properties described below. Suppose the formal parameters to
\ourmul are denoted $P_{in}$ and $Q_{in}$ respectively. We
assume that the loop counter $i$ is initialized to 0 and is
incremented by 1 at the end of each iteration up to the bit
length $n$ of the input tnums. When we say ``the $k^{th}$
iteration'' for some fixed $k$, we mean the loop iteration
where the value of the loop counter $i$ {\em at the end of
the loop body} is $k$,  $1 \leq k \leq n$.

\proofheading{Property P1. $P$ and $Q$ are bit-shifted
versions of $P_{in}$ and $Q_{in}$.} At the end of the
$k^{th}$ iteration, we have that $P =
\Tnumrshift{P_{in}}{k}$ and $Q = \Tnumlshift{Q_{in}}{k}$.
Further, {\em within} the body of iteration $k$ of the loop,
$P = \Tnumrshift{P_{in}}{k-1}$ and $Q =
\Tnumlshift{Q_{in}}{k-1}$. In particular, within the body of
the $k^{th}$ iteration, $\ithbitx{P.v}{0} =
\ithbitx{P_{in}.v}{k-1}$ and $\ithbitx{P.m}{0} =
\ithbitx{P_{in}.m}{k-1}$, $Q.v = Q_{in}.v \,\bitwiseLshift\,
(k-1)$ and $Q.m = Q_{in}.m \,\bitwiseLshift\, (k-1)$. 

\begin{proof}  The algorithm only modifies $P$ and $Q$ via
tnum bitshifting at the end of each iteration. It is
straightforward to show the properties above by induction on
the loop counter $i$ and the semantics of the tnum bit
shifting algorithms.
\end{proof}

\proofheading{Property P2. $ACC_V$ and $ACC_M$ are
value-mask-decomposed summations of partial products.} There
exist tnums $T_0, T_1, \ldots, T_{n-1}$ such that (i) any
concrete $j^{th}$ partial product, $z_j \triangleq
\ithbitx{x}{j} * {y} \,\bitwiseLshift\, {j}$ is a member of
$\gamma(T_j)$, \ie $\Tnumin{z_j}{T_j}$ for $0 \leq j \leq
n-1$; (ii) at the end of the $\nth{k}$ iteration of the
loop, $ACC_V = \code{tnum_add}_{j=0}^{k-1} tnum(T_j.v, 0)$,
and (iii) at the end of the $\nth{k}$ iteration of the loop,
$ACC_M = \code{tnum_add}_{j=0}^{k-1} tnum(0, T_j.m)$.

\begin{proof}
Consider the $j^{th}$ concrete partial product $z_j
\triangleq \ithbitx{x}{j} * {y}\,\bitwiseLshift\,{j}$, where
$ \Tnumin{P_{in}}{y}$,  $\Tnumin{Q_{in}}{0}$, and  $\leq j
\leq n-1$. In the algorithm, the consideration of the
concrete $j^{th}$ partial product $z_j$ occurs within the
code body of the $j^{th}$ iteration, where $k = j + 1$. By
Property P1 and the soundness of tnum bitshifting, at the
beginning of the $k^{th}$ iteration, $\Tnumin{y
\,\bitwiseLshift\, j}{Q}$. Further, by Property P1, the
soundness of tnum bitshifting, and the definition of
$\gamma$, it is also the case that $\ithbitx{x}{j}
\;\bitwiseAnd\; \bitwiseNot{\ithbitx{P.m}{0}} =
\ithbitx{P.v}{0}$.

\proofsubheading{(i) Construction of $T_j$.} Given
$\Tnumin{x}{P}$, there are three possibilities for the bit
$\ithbitx{x}{j}$. 

\begin{enumerate}
\item $\ithbitx{x}{j}$ is a known 1: Suppose
$\ithbitx{x}{j}$ is known to be 1 (\ie $\ithbitx{x}{j} = 1$,
$\ithbitx{P_{in}.v}{j} = 1$, and $\ithbitx{P_{in}.m}{j} =
0$). Then, the concrete partial product $z_j =
y\,\bitwiseLshift\,j$. Since
$\Tnumin{y\,\bitwiseLshift\,j}{Q}$, we set $T_{j} \triangleq
Q$. Hence, $\Tnumin{z_j}{T_j}$.
\item $\ithbitx{x}{j}$ is unknown: Suppose $\ithbitx{x}{j}$
is uncertain (\ie $\ithbitx{P_{in}.m}{j} = 1$ and
$\ithbitx{P_{in}.v}{j} = 0$). The partial product $z_j$ can
take on a value $0$ if $\ithbitx{x}{j} = 0$, else the value
of the partial product is $y \,\bitwiseLshift\, j$. Since
$\Tnumin{y\,\bitwiseLshift\,j}{Q}$, we can construct $T_j$
so that it contains the concrete value $0$ as well as all
concrete values in $Q$. Invoking
\Lemma{lemma:union-tnum-with-zero-appendix}, setting $T_j
\triangleq tnum(0, Q.v \mid Q.m)$ will ensure that
$\Tnumin{z_j}{T_j}$.
\item $\ithbitx{x}{j}$ is a known 0:  Suppose
$\ithbitx{x}{j}$ is known to be 0 (\ie $\ithbitx{x}{j} = 0$,
$\ithbitx{P_{in}.v}{j} = 0$, and $\ithbitx{P_{in}.m}{j} =
0$). Then, the partial product $z_j = 0$. We set $T_j
\triangleq 0$. Hence, $\Tnumin{z_j}{T_j}$.
\end{enumerate}

\proofsubheading{(ii) $ACC_V = \code{tnum_add}_{j=0}^{i-1}
tnum(T_j.v, 0)$.} Within the body of the loop, $ACC_V$ is
only updated using tnum addition operations, and in at most
one line of code per loop iteration. Further, it can be seen
easily that this update $ACC_V$ occurs using the value
component of the corresponding partial product tnum $T_j$
(including the possibility of making no updates to $ACC_V$
when $T_j.v = 0)$. It is straightforward to construct an
inductive proof that $ACC_V = \code{tnum_add}_{j=0}^{k-1}
tnum(T_j.v, 0)$ after the $k^{th}$ iteration.

\proofsubheading{(iii) $ACC_M = \code{tnum_add}_{j=0}^{i-1}
tnum(0, T_j.m)$.} The proof is very similar to Property
P2 (ii) above.

\end{proof}

\vspace{10pt}
\proofheading{Property P3 (product containment).} At the end
of the $n^{th}$ iteration, $\code{tnum_add}(ACC_V, ACC_M)$
contains the sum of all concrete partial products $z_j$ for
$j = 0, \ldots, n-1$. That is, $\code{tnum_add}(ACC_V,
ACC_M)$ contains all concrete values of $x * y$.

\begin{proof}
We apply
\Lemma{lemma:value-mask-decoupled-additions-appendix}, where
the constituent tnums are $T_0, T_1, \ldots, T_{n-1}$ and
the concrete values considered are the possible partial
products $z_j$. Combined with Property P2, we have that for
all possible partial products $z_j$,
$\Tnumin{\sum_{j=0}^{n-1} z_j}{\code{tnum_add}(ACC_V,
ACC_M)}$. Further, from \Obs{partial-products-appendix}, we
know that $\sum_{j=0}^{n-1} z_j = \sum_{j=0}^{n-1}
\ithbitx{x}{j} * (y \,\bitwiseLshift\, j) = x * y$. Hence,
$\code{tnum_add}(ACC_V, ACC_M)$ contains all concrete
products $x * y$.
\end{proof}

Property P3 concludes the proof of soundness of \ourmul. The
returned tnum $R$ contains all products of concrete values
$x * y$ where $\Tnumin{x}{P}$ and $\Tnumin{y}{Q}$.

\end{proof}

\subsection{Proof for Tnum Subtraction}
\begin{lstlisting}[
caption={ {\footnotesize Linux kernel's implementation of tnum subtraction
(\code{tnum_sub})}}, label={lst:tnum_sub}]
def tnum_sub(tnum P, tnum Q):

	u64 dv := P.v - Q.v
	u64 alpha := dv + P.m
	u64 beta := dv - Q.m
	u64 chi := alpha (*$\bitwiseXor$*) beta
	u64 eta := chi | P.m | Q.m	
	tnum R := tnum(dv & ~eta, eta)	
	return R
\end{lstlisting}

\begin{theorem}
Tnum subtraction as defined by the linux kernel \Lst{tnum_sub} is
sound and optimal. 
\end{theorem}

\begin{proof}
We first look at the equation for the subtraction of two
binary numbers using the full subtractor equation. 

\begin{definition}
\label{subtraction of bits}
When subtracting two concrete binary numbers p and q, each
bit of the subtraction result $r$ is set according to the
following:
$$r[i] = p[i] \; \bitwiseXor \; q[i] \; \bitwiseXor \;
\ithbit{b}{{in}}{i}$$

where $\; \bitwiseXor \;$ is the exclusive-or operation and
$\ithbit{b}{{in}}{i} = \ithbit{b}{{out}}{i-1}$ and
$\ithbit{b}{{out}}{i-1}$ is the borrow out from the
subtraction in bit position $i-1$. The borrow out bit at the
$i^{th}$ position is defined by the full subtractor
equation:

$$\ithbit{b}{{out}}{i} = (\; \bitwiseNot \; p[i] \;
\bitwiseAnd \; q[i]) \; \bitwiseOr \; (\ithbit{b}{{in}}{i}
\; \bitwiseAnd \; \; \bitwiseNot \;(p[i] \; \bitwiseXor \;
q[i]))$$
\end{definition}

At a given bit position in the tnum, the result of a tnum
subtraction will be uncertain if either of the operand bits
$\ithbitx{p}{i}$ or $\ithbitx{q}{i}$ is uncertain, or if the
borrow-in bit $\ithbitx{b_{in}}{i}$ generated from
more-significant bit positions is uncertain. The crux of the
proof lies in identifying when borrow-in bits are uncertain,
by distinguishing the borrows generated due to the
uncertainty in the operands from the borrows present in any
concrete subtraction from the input tnums. 

Given two tnums $P$ and $Q$, the uncertainty contributed by
the operands to the result is $P.v \; \bitwiseOr \; Q.v$. To
account for uncertainty that arises in
$\ithbit{b}{{in}}{i}$, we need to consider the sequence of
borrows that can result from the subtraction. Let $p$ and
$q$ be any concrete value in tnum $P$ and tnum $Q$,
respectively.
      
Suppose $p$ and $q$ are two concrete values in tnum $P$ and
tnum $Q$, respectively, i.e., $p \in \gamma(P ), q \in
\gamma(Q)$. 

\begin{lemma}
\label{lemma:minimum-borrows-lemma-appendix}
\textbf{Minimum Borrows Lemma.}
The subtraction $\alpha = (P.v + P.m)- Q.v$
will produce the sequence of borrow bits with the least
number of $1$s out of all possible  subtractions $p-q$.
\end{lemma}

\begin{proof}        	
Let $\ithbit{b}{{out}}{i}$ denote the borrow out bit at
the $i$th position produced by $p - q$ and
$\ithbit{b}{\alpha}{i}$ denote the borrow out bit at the
$i^{th}$ position produced by $(P.v + P.m)- Q.v$. We will
prove by induction that if $(P.v[i] + P.m[i])- Q.v[i]$
produces a borrow bit then so must $p[i] - q[i]$:

\textbf{Base case}: At bit position $i=0$, $(P.v[i] +
P.m[i])- Q.v[i]$ produces a borrow out bit only when $P.v[i]
+ P.m[i]=0$ and $Q.v[i]=1$. Since a certain bit has the same
value across all members of the set then by
\Eqn{eqn:tnum_value_corresp} $p[i]=0$ and $q[i]=1$, which
must also produce a borrow bit. Hence,
$\ithbit{b}{\alpha}{i} = 1 \rightarrow \ithbit{b}{{out}}{i}
= 1$ holds.  

\textbf{Induction step}: Assume that for all $0 \leq j \leq
i-1$, $\ithbit{b}{{out}}{j} \geq
\ithbit{b}{\alpha}{j}$. Now, we show that
$\ithbit{b}{{out}}{i} \geq \ithbit{b}{\alpha}{i}$ by
considering all possible cases:

\begin{enumerate}
\item $P.v[i] + P.m[i]=1$ and $Q.v[i] = 0$.
$\ithbit{b}{\alpha}{i}$ will never produce a borrow which
means $\ithbit{b}{\alpha}{i} = 1 \rightarrow
\ithbit{b}{{out}}{i} = 1$ holds.  

\item $P.v[i] + P.m[i] = 0$ and $Q.v[i] = 1$.
$\ithbit{b}{\alpha}{i}$ will always produce a borrow bit.
Since a certain bit has the same value across all members of
the set we know that by \Eqn{eqn:tnum_value_corresp},
$p[i]=0$ and $q[i]=1$, which will also produce a borrow bit.
This means $\ithbit{b}{\alpha}{i} = 1 \rightarrow
\ithbit{b}{{out}}{i} = 1$ holds.  

\item $P.v[i] + P.m[i] = 0$ and $Q.v[i] = 0$.
$\ithbit{b}{\alpha}{i}$ will produce a borrow only when
position $i-1$ produces a borrow. This case implies that
$p[i]=0$ and $q[i]=0$ by \Eqn{eqn:tnum_value_corresp}, which
will produce a borrow only when bit position $i-1$ produces
a borrow. Thus, by the induction hypothesis,
$\ithbit{b}{\alpha}{i} = 1 \rightarrow \ithbit{b}{{out}}{i}
= 1$ holds.  

\item $P.v[i] + P.m[i] = 1$ and $Q.v[i] = 1$. This case
entails two possibilites as follows:

\begin{enumerate}
\item $\ithbit{b}{\alpha}{i-1} = 0$. This implies that
  $\ithbit{b}{\alpha}{i} = 0$ which means it will not
  produce a borrow bit. Hence, $\ithbit{b}{\alpha}{i} = 1
  \rightarrow \ithbit{b}{{out}}{i} = 1$ holds.  
  
\item $\ithbit{b}{\alpha}{i-1} = 1$.
$\ithbit{b}{\alpha}{i}$ will produce a borrow only when
position $i-1$ produces a borrow. This case implies that
$p[i]=0 \vee 1$ and $q[i]=0$ by
\Eqn{eqn:tnum_value_corresp}, which will produce a borrow
only when bit position $i-1$ produces a borrow. Thus, by the
induction hypothesis, $\ithbit{b}{\alpha}{i} = 1 \rightarrow
\ithbit{b}{{out}}{i} = 1$ holds.  
\end{enumerate}
\end{enumerate}

Hence, it follows that if $P.v[i] + Q.v[i]$ produces a
borrow bit then so must $p[i] + q[i]$.
\end{proof}

\begin{lemma}
\label{lemma:maximum-borrows-lemma-appendix}
\textbf{Maximum Borrows Lemma.}
The subtraction $\beta = P.v - (Q.v + Q.m)$ will produce the
sequence of borrow bits with the most number of 1s out of
all possible subtractions $p$ - $q$.
\end{lemma}

We prove that  That is, any $p$ - $q$ will produce a
sequence of borrow bits with 1s in at most those positions
where $\beta$ produced borrow bits set to 1.
        
\begin{proof}

Let $\ithbit{b}{{out}}{i}$ denote the borrow out bit at
the $i$th position produced by $p - q$ and
$\ithbit{b}{\beta}{i}$ denote the borrow out bit at the
$i^{th}$ position produced by $P.v - (Q.v + Q.m)$. We will
prove by induction that if $P.v - (Q.v + Q.m)$ does not
produce a borrow bit then neither will $p[i] - q[i]$:

\textbf{Base case}: At bit position $i=0$, a concrete
subtraction $p[0] - q[0]$ produces a borrow bit only when
$p[0] = 0$ and $q[0] = 1$. If $p[0] = 0$ then it must mean
that $P.v[0] = 0$ (regardless if the bit is certain or not).
Similarly, $q[0] =1$ implies $Q.v[0] =1 $ (when the bit is
certain). Otherwise, $Q.m[0] = 1$. Thus, $P.v - (Q.v + Q.m)$
will produce a borrow bit whenver any $p[0] + q[0]$ produces
a borrow, and not more. Hence, $\ithbit{b}{\beta}{i} = 0
\rightarrow \ithbit{b}{{out}}{i} = 0$ holds.  

\textbf{Induction step}: Assume that for all $0 \leq j \leq
i-1$, $\ithbit{b}{\beta}{j} \geq \ithbit{b}{{out}}{j}$. Now,
we show that $\ithbit{b}{\beta}{i} = 0 \rightarrow
\ithbit{b}{{out}}{i} = 0$ holds.  by considering all
possible cases:

\begin{enumerate}
\item $P.v[i] =1$ and $Q.v[i] + Q.m[i]=0$. This will never
produce a borrow bit. Since $P.v=1$ implies $p=1$ and
$Q.v[i] + Q.m[i]=0$ implies $q=0$ (by
\Eqn{eqn:tnum_value_corresp}) then $p[i]-q[i]$ will not
produce a borrow bit either. Hence, $\ithbit{b}{\beta}{i} =
0 \rightarrow \ithbit{b}{{out}}{i} = 0$ holds.  

\item $P.v[i] = 0$ and $Q.v[i] + Q.m[i]= 1$.  This will
always produce a borrow bit which means,
$\ithbit{b}{\beta}{i} = 0 \rightarrow \ithbit{b}{{out}}{i} =
0$ holds vacuously.  

\item $P.v[i] = 0$ and $Q.v[i] + Q.m[i]= 0$.
$\ithbit{b}{\beta}{i}$ will not produce a borrow only when
position $i-1$ does not produce a borrow. This case implies
that $p[i]=0$ and $q[i]=0$ by \Eqn{eqn:tnum_value_corresp},
which will not produce a borrow only when bit position $i-1$
does not produces a borrow. Thus, by the induction hypothesis,
$\ithbit{b}{\beta}{i} = 0 \rightarrow \ithbit{b}{{out}}{i}
= 0$ holds.  

\item $P.v[i] = 1$ and $Q.v[i] + Q.m[i]=1$. This case
entails two possibilites as follows:

\begin{enumerate}
\item $\ithbit{b}{\beta}{i-1} = 0$. This implies that
$\ithbit{b}{\beta}{i} = 0$ which means it will not
produce a borrow bit. Hence, $\ithbit{b}{\beta}{i} = 0
\rightarrow \ithbit{b}{{out}}{i} = 0$ holds.  

\item $\ithbit{b}{\beta}{i-1} = 1$. $\ithbit{b}{\beta}{i}$
will not produce a borrow only when position $i-1$ does not
produces a borrow. This case implies that $p[i]=1 $ and
$q[i]=0 \vee 1$ by \Eqn{eqn:tnum_value_corresp}, which will
not produce a borrow only when bit position $i-1$ does not
produce a borrow. Thus, by the induction hypothesis,
$\ithbit{b}{\beta}{i} = 0 \rightarrow \ithbit{b}{{out}}{i}
= 0$ holds.  
\end{enumerate}
\end{enumerate}

Hence, if $P.v[i]$ - $(Q.v[i] +Q.m[i])$ does not produce a
borrow bit then neither will concrete subtraction $p[i] -
q[i]$.

\end{proof}
      
\begin{lemma}
\label{lemma:capture-uncertainty-lemma-sub-appendix}
\textbf{Capture uncertainty lemma.} 
 Let $\alpha_{b}$ and $\beta_{b}$ be the sequence of
borrow-in bits from the subtractions in $\alpha$ and
$\beta$, respectively. Suppose $\chi_b \triangleq \alpha_{b}
\; \bitwiseXor \; \beta_{b}$. The bit positions $k$ where
$\ithbitx{\chi_c}{k} = 0$ have borrow bits fixed in all
concrete subtractions $p-q$ from $-(\gamma(P), \gamma(Q))$.
The bit positions $k$ where $\ithbitx{\chi_c}{k} = 1$ vary
depending on the concrete subtraction: \ie $\exists p_1, p_2
\in \gamma(P), q_1, q_2 \in \gamma(Q)$ such that $p_1 - q_1$
has its borrow bit set at position $k$ but $p_2 - q_2$ has
that bit unset.
\end{lemma}

Intuitively, from the minimum borrows lemma, any borrow bit
that is produced by $\alpha_{b}$ must also be produced by
any concrete subtraction $p-q$ because $\alpha_{b}$
contains the minimal amount of borrow bits. From the maximum
borrows lemma, no concrete subtraction $p - q$ may have
more borrow bits than $\beta_{b}$. Hence, $\alpha_{b} \;
\bitwiseXor \; \beta_{b}$ represents the borrows that arise
purely from the uncertainty in the concrete operands picked
from  $P$ and $Q$.

\textbf{Soundness and optimality.} Tnum subtraction
(\code{tnum_sub}) uses $(\alpha \; \bitwiseXor \; \beta) \;
\bitwiseOr \; P.m \; \bitwiseOr \; Q.m$ to compute the
uncertainty of the result. To show soundness, we prove that
$(\alpha \; \bitwiseXor \; \beta) \; \bitwiseOr \; P.m \;
\bitwiseOr \; Q.m$ and $(\alpha_{b} \; \bitwiseXor \;
\beta_{b}) \; \bitwiseOr \; P.m \; \bitwiseOr \; Q.m$
compute the same result.

From Definition~ \ref{subtraction of bits}, we have $r[i] =
p[i] \; \bitwiseXor \; q[i] \; \bitwiseXor \;
\ithbit{b}{{in}}{i}$. Hence, the $i^{th}$-bit of $\alpha[i]$
is $(P.v[i]+P.m[i])\; \bitwiseXor \; Q.v[i] \; \bitwiseXor
\; \alpha_b[i]$. Similarly, $i^{th}$-bit of $\beta[i]$ is
$P.v[i] \; \bitwiseXor \; (Q.v[i] + Q.m[i]) \; \bitwiseXor
\; \beta_b[i]$.

\begin{small}
\begin{equation*}
\begin{aligned}
\alpha[i] \; \bitwiseXor \; \beta[i] = & 
(P.v[i]+P.m[i])\; \bitwiseXor \;
Q.v[i] \; \bitwiseXor \; \alpha_b[i]
\\ & \; \bitwiseXor \; P.v[i] \; \bitwiseXor \; 
(Q.v[i] + Q.m[i]) \; \bitwiseXor \; \beta_b[i]
\end{aligned}
\end{equation*}
\end{small}

To identify where $\alpha \; \bitwiseXor \; \beta$ can
differ from $\alpha_{b} \; \bitwiseXor \; \beta_{b}$,
consider the exclusive-or of them (\ie $(\alpha \;
\bitwiseXor \; \beta) \; \bitwiseXor \; (\alpha_{b} \;
\bitwiseXor \; \beta_{b})$). If they are identical, the
result is 0. They are different when $(\alpha \; \bitwiseXor
\; \beta) \; \bitwiseXor \; (\alpha_{b} \; \bitwiseXor \;
\beta_{b})$ is 1.

\begin{small}
\begin{equation*}
\begin{aligned}
& (\alpha[i] \; \bitwiseXor \; \beta[i]) \; \bitwiseXor \; 
(\alpha_{b}[i] \; \bitwiseXor \; \beta_{b}[i]) \\
& = (P.v[i]+P.m[i])\; \bitwiseXor \;
Q.v[i] \; \bitwiseXor \; \alpha_b[i]
\; \bitwiseXor \; P.v[i] \; \bitwiseXor \; 
\\ & \;\;\;(Q.v[i] + Q.m[i]) \; \bitwiseXor \; \beta_b[i] \; 
\bitwiseXor \; \alpha_b[i] \; \bitwiseXor \; \beta_b[i] \\
&= P.v[i] \; \bitwiseXor \; (P.v[i] + P.m[i]) \; \bitwiseXor 
\; Q.v[i] \; \bitwiseXor \; (Q.v[i] + Q.m[i]) 
\end{aligned}
\end{equation*}
\end{small}
      
Now consider the sub-term in the above formula $P.v[i] \;
\bitwiseXor \; (P.v[i] + P.m[i])$. For a non-empty tnum
$P.v[i] \; \bitwiseAnd \; P.m[i] = 0$. Hence the term
$P.v[i] \; \bitwiseXor \; (P.v[i] + P.m[i])$ is equivalent
to $P.m[i]$. Similarly, term $Q.v[i] \; \bitwiseXor \;
(Q.v[i] + Q.m[i])$ is equivalent to $Q.m[i]$.

Thus, $\alpha[i] \,\bitwiseXor\, \beta[i]$ differs from
$\alpha_{c}[i] \,\bitwiseXor\, \beta_{c}[i]$ when $P.m[i]
\,\bitwiseXor\, Q.m[i] = 1$. Now, if $P.m[i] \,\bitwiseXor\,
Q.m[i] = 1$ then $P.m[i] \,\bitwiseOr \,Q.m[i] = 1$ as well.
Thus, setting $x \triangleq \alpha \,\bitwiseXor\, \beta$ or $x
\triangleq \alpha_{c} \,\bitwiseXor\, \beta_{c}$ is immaterial
to the value $x \,\bitwiseOr\, P.m \,\bitwiseOr\, Q.m.$.
\end{proof}

\begin{lemma}
\label{thm:soundness-optimality-of-tnum-sub-appendix}
\textbf{Soundness and optimality of tnum\_sub.} 
The algorithm \code{tnum_sub} shown in \Lst{tnum_add} is a
sound and optimal abstraction of concrete subtraction over
\nbit bitvectors.
\end{lemma}

\begin{proof}
Since \code{tnum_sub} captures all, and only, the
uncertainty in the concrete results of tnum subtraction, it
is sound and optimal.    
\end{proof}
\subsection{Automated verification of tnum operations using SMT solvers}

We previously discussed our SMT encoding of the
implementation of tnum operators using theory of fixed-size
bitvectors in \Sec{automated-verification}. We use the
theory of fixed-size bitvectors, and encode the
implementation using the python bindings of the Z3 SMT
solver~\cite{z3-solver}. The size of the bitvector is a
parameter for performing bounded automated verification.
Given that the Linux kernel uses machine arithmetic with
64-bits, we use bitvectors of width $64$ wherever feasible.

\Para{Encoding abstract tnum multiplication}. The Linux
kernel's implementation for abstract multiplication of
tnums, \kernmul, is shown in \Lst{kern_mul}. Among the
kernel's abstract operators, multiplication is the most
challenging to verify automatically, because it involves a
call to a function \code{hma} that contains a loop.
Fortunately, the loop test is simple to encode (the loop
runs at most $n$ times for \nbit bitvectors), hence it is
possible to unroll the loop and rewrite the code in static
single assignment
form~\cite{compiler-implementation-appel04}. The rest of the
encoding of multiplication is straightforward. The details
follow.

Recall that we defined the following predicates: (1) the
$\memberPredicate{}$ predicate which asserts for a concrete
value $x$ and tnum $P$ that $x \in \gamma(P)$, and (2) the
$\wellformedPredicate$ predicate, which allows quantifying
over non-empty tnums $P$. 

\begin{small}
\begin{align}
\begin{split}
\label{eqn_mem_wff_appendix}  
\memberPredicate(x, P) \; \triangleq \; x \; \bitwiseAnd \; \bitwiseNot \tmi{P} = \tvi{P}\\
\wellformedPredicate(P) \triangleq \, \tvi{P} \, \bitwiseAnd \, \tmi{P} = 0
\end{split}
\end{align}
\end{small}

The soundness predicate for a given pair of abstract and
concrete operators $\opT, \opC$ is

\begin{small}
\begin{align}
\begin{split}
\label{eqn_soundness_2_appendix} 
\forall & P, Q \in \bb{T}{n}, x, y \in \bb{Z}{n} \; : 
\\ & \wellformedPredicate(P) \; \wedge \wellformedPredicate(Q) \; \wedge \memberPredicate(x, P) 
\\ & \; \wedge \memberPredicate(y, Q) \; \wedge z = \opC(x, y) \; \wedge R = \opT(P, Q) \;
\\ & \Rightarrow \memberPredicate(z, R)
\end{split}
\end{align}
\end{small}

Equation \eqref{eqn_hma_appendix} below shows the predicate
$hma$ which encodes the \code{hma} function from
\Lst{kern_mul}, with the loop shown unrolled 2 times. The
$ite$ term encodes \emph{if-then-else} and $add$ is the
predicate from \Eqn{eqn_add_predicate}. 

\begin{small}
\begin{align}
\begin{split}
\label{eqn_hma_appendix}
hma&(ACC_{in}, x_{in}, y_{in}, R) \triangleq \\
 & (ACC_{0} = ACC_{in}) \land (x_{0} = x_{in}) \land (y_{0} = y_{in}) \\
 & \land  ite(y_{0}  \, \bitwiseAnd \,  1  = 1, \; \\ 
 & \; \; \; \; \; \;  add(ACC_{0}, tnum(0, x_{0}), ACC_{1}), \; \\
 & \; \; \; \; \; \; ACC_{1} = ACC_{0}) \\
 & \land (y_{1} = y_{0} \, \bitwiseRshift \, 1) \land (x_{1} = x_{0} \, \bitwiseLshift\,  1) \\
 & \land  ite(y_{1} \,  \, \bitwiseAnd \, 1 = 1,\; \\
 & \; \; \; \; \; \;  add(ACC_{1}, tnum(0, x_{1}), ACC_{2}), \; \\
 & \; \; \; \; \; \;  ACC_{2} = ACC_{1}) \\
 & \land (y_{2} = y_{1} \, \bitwiseRshift \, 1) \land (x_{2} = x_{1} \, \bitwiseLshift\,  1) \\
 & \ldots \\
 & \land R = ACC_{64}
\end{split}
\end{align}
\end{small}

\Para{Spot-checking the correctness of our SMT encodings.}
To ensure that our encodings of the tnum operators in
first-order logic are accurate, we developed a test harness.
Consider a tnum abstract operator $\code{op} (tnum \; P,
tnum \; Q)$ with a first order logic formula $f_{op}$
encoding its action symbolically. $f_{op}$ contains
bitvector variables corresponding to values and masks for
input tnums $P$ and $Q$ and output tnum $R$. We tested our
encoding as follows. First, we produced many random input
tnum pairs ($X$, $Y$). Next, for each pair ($X$, $Y$), we
executed the C code of the kernel tnum operator to produce
an output tnum $Z$. Finally, we use Z3 ~\cite{z3-solver} and
obtain a model for the formula $P = X \wedge Q = Y \wedge
f_{op}$. The test passes only if the formula above is
satisfiable and its model interpretation of output $R$ is
the same as tnum $Z$ from the C execution.

\subsection{Comparing the precision of \ourmul versus
\kernmul with increasing bitwidth}

\Para{Setup}. We evaluate the precision of \code{our_mul}
compared to \reghermul and \kernmul with increasing
bitwidth. For a particular bitwidth, we provide as input all
possible tnum pairs at that bitwidth to both these
algorithms. We perform this experiments for bitwidths $5$
through $10$ (we stop at $n=10$ to keep running times
tractable). We then try to answer the questions: What
percentage of input pairs lead to a better precision in the
result produced by \ourmul? How does this percentage change
with increasing bitwidth?

\Para{Results.} We make the following observations (see
Table~\ref{table:muls_precision_increasing_bitwidth}). (1)
The percentage of inputs where \ourmul produces a different
output that \kernmul increases. (2) Output tnums are always
comparable for bitwidths $n = 5$ though $ n = 8$. For higher
bitwidths, although less than $100\%$ of differing output
tnums are comparable, a high fraction ($>99\%$) still remain
comparable. (3) The tnums produced by \ourmul are more
precise than \kernmul for a higher fraction of the
comparable output tnums. (4) Considering comparable output
tnums for a given input, the fraction of output tnums
produced by \ourmul which are more precise than \kernmul
increases as bitwidth increases. This trend is quite
remarkable: it goes to show that \ourmul produces
increasingly more precise output tnums than \kernmul as the
bitwidth ($n$) of the tnums increases.

\begin{center}
\captionsetup{type=figure}
\includegraphics[width=\columnwidth]{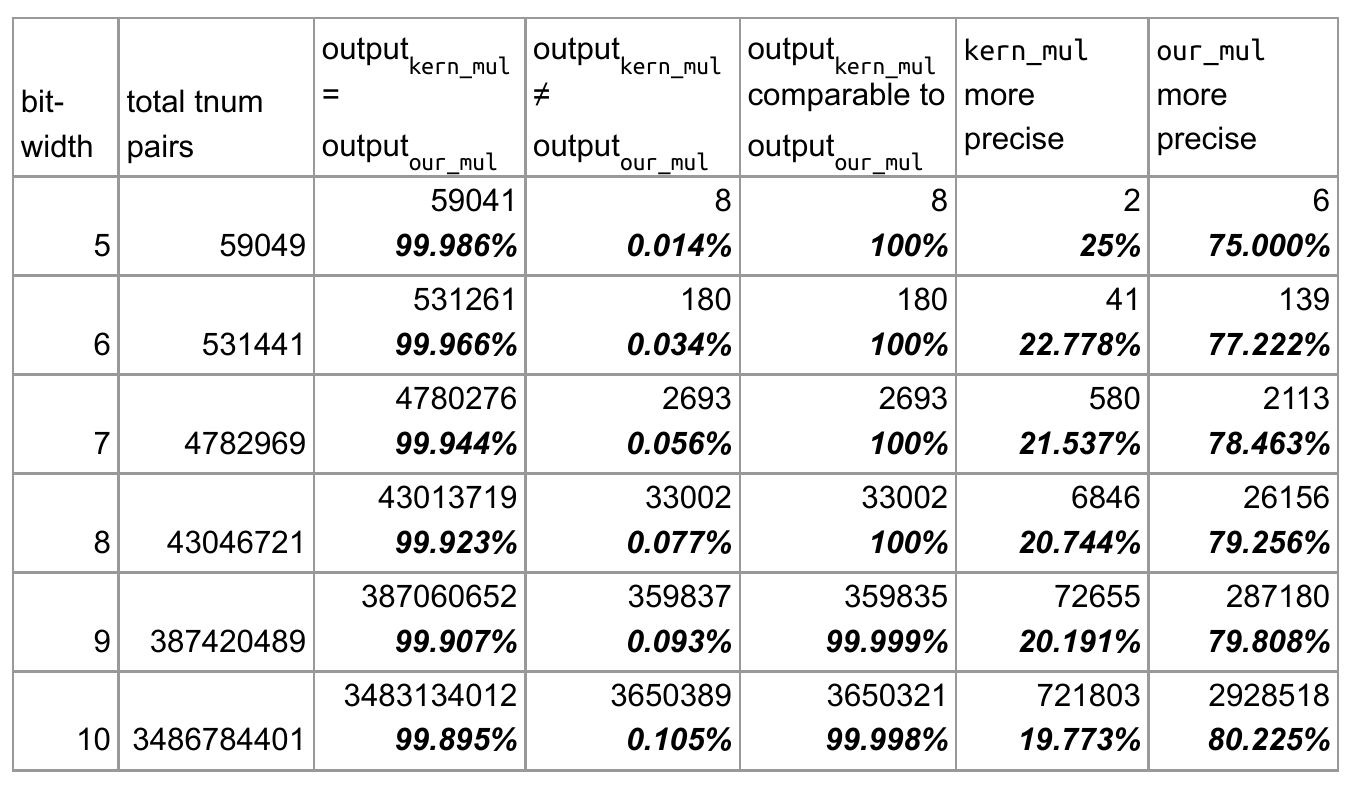}
\captionof{table}{\footnotesize{Table comparing outputs from
\ourmul and \kernmul, when they are given inputs drawn from
all possible tnum pairs of bit of bitwidth $5$ through $10$.
The table shows ($L - R$, starting from column $3$):  (i)
Percentage of inputs where outpts from both the algorithms
are exactly the same, (ii) Percentage of inputs where
outputs from both algorithms differ (iii) Of those outputs
which differ, the percentage of outputs which are
comparable, (iv) Of the outputs that differ but are still
comparable, the percentage where \kernmul is more precise
than \ourmul, (v) Of the outputs that differ but are still
comparable, the percentage where \ourmul is more precise
than \kernmul.}}
\label{table:muls_precision_increasing_bitwidth}
\end{center}

\subsection{Proof of the Galios connection} 

\begin{theorem}
\label{galios-connection-proof}
\GaliosConnectionIntegerToBitField is a Galios connection.
\end{theorem}

We prove that our abstraction function $\alpha:
2^{\mathbb{Z}_n} \rightarrow \mathbb{Z}_n \times
\mathbb{Z}_n$, and the concretization function $\gamma:
\mathbb{Z}_n \times \mathbb{Z}_n \rightarrow
2^{\mathbb{Z}_n}$ form a Galois connection by proving the
following properties:
\begin{enumerate}
\item $\gamma$ is monotonic
\item $\alpha$ is monotonic
\item $\gamma \circ \alpha$ is extensive
\item $\alpha \circ \gamma$ is reductive
\end{enumerate}

\proofheading{Property \textnormal{G1}. $\alpha$ is monotonic, 
\ie $\forall C_{1}, C_{2} \in \Zpowerset, 
C_{1} \leqconc C_{2} \rightarrow \alpha(C_{1}) \leqabst \alpha(C_{2})
$}. 
\begin{proof}
We have $C_{1} \subseteq C_{2}$. Hence, by the definition of
$\alpha$ in \Eqn{eqn_alpha}, the following properties hold:

\begin{align}
\begin{split}
\label{eqn_alpha_c_props}
\alpha_{\bitwiseAnd}(C_{2}) = \alpha_{\bitwiseAnd}(C_{1})
\; \bitwiseAnd \; \alpha_{\bitwiseAnd}(C_{2} - C{1})  \\
\alpha_{\bitwiseOr}(C_{2}) =
\alpha_{\bitwiseOr}(C_{1}) \; \bitwiseOr \;
\alpha_{\bitwiseOr}(C_{2} - C{1}) \\
\forall k \in {0, 1}: \alpha_{\bitwiseOr}(C_{2)} = k 
\leftrightarrow \alpha_{\bitwiseOr}(C_{1)} = k \\
\forall k \in {0, 1}: \alpha_{\bitwiseAnd}(C_{2)} = k 
\leftrightarrow \alpha_{\bitwiseAnd}(C_{1)} = k
\end{split}
\end{align}
Here, $C_{2} - C_{1}$ is the set difference, which includes
all the elements in $C_{2}$ that are not in $C_{1}$. Now,
For a concrete set $C$ and for some $c \in C$ we use the
notation $c[i]$ to denote the \ith{i} bit of $c$. Similarly,
for a tnum $T$, we use the notation $T[i]$ to denote the
\ith{i} trit of $T$. 

We have to prove that $\alpha(C_{1}) \leqabst
\alpha(C_{2})$.  By the definition of the partial order of
the abstract domain, this means: 

\begin{align}
\begin{split}
\label{eqn_gamma_monotonic_to_prove}
\forall i, 0 \leq i \leq n-1, \forall k
\in \{0, 1\}: \\ 
(\ithbitx{\alpha(C_{1})}{i} = \mu \Rightarrow
\ithbitx{\alpha(C_{2})}{i} = \mu) \wedge \\
(\ithbitx{\alpha(C_{2})}{i} = k \Rightarrow 
\ithbitx{\alpha(C_{1})}{i} = k). 
\end{split}
\end{align}

Take a particular bit position $i$. We consider each clause
of the conjunction one after the other.

\begin{enumerate}
\item $\alpha(C_{1})[i] = \mu = tnum(0, 1)$. This implies,
that $\alpha_{\bitwiseAnd}(C_{1}[i] = 0$ and
$\alpha_{\bitwiseOr}(C_{1}[i] = 1$. It follows from
\Eqn{eqn_alpha_c_props} that $\alpha_{\bitwiseAnd}(C_{2}[i]
= 0$ and $\alpha_{\bitwiseOr}(C_{2}[i] = 1$. Hence
$\alpha(C_{2}[i]) = tnum(0, 1) = \mu$. Hence,
$\alpha(C_{1}[i] = \mu \rightarrow \alpha(C_{2}[i] = \mu$.
\item We consider two cases:
\begin{enumerate}
\item $\alpha(C_{2})[i] = 0 = tnum(0, 0)$. This implies,
that $\alpha_{\bitwiseAnd}(C_{2}[i] = 0$ and
$\alpha_{\bitwiseOr}(C_{2}[i] = 0$. It follows from
\Eqn{eqn_alpha_c_props} that $\alpha_{\bitwiseAnd}(C_{1}[i]
= 0$ and $\alpha_{\bitwiseOr}(C_{2}[i] = 0$. Hence
$\alpha(C_{1}[i]) = tnum(0, 0) = 0$. Hence,
$\alpha(C_{2}[i] = 0 \rightarrow \alpha(C_{1}[i] = 0$.    
\item $\alpha(C_{2})[i] = 1 = tnum(1, 0)$. This implies,
that $\alpha_{\bitwiseAnd}(C_{2}[i] = 1$ and
$\alpha_{\bitwiseOr}(C_{2}[i] = 1$. It follows from
\Eqn{eqn_alpha_c_props} that $\alpha_{\bitwiseAnd}(C_{1}[i]
= 1$ and $\alpha_{\bitwiseOr}(C_{2}[i] = 1$. Hence
$\alpha(C_{1}[i]) = tnum(1, 0) = 0$. Hence,
$\alpha(C_{2}[i] = 1 \rightarrow \alpha(C_{1}[i] = 1$.
\end{enumerate}
From both the above cases, we get: $\forall k \in {0, 1}, 
\alpha(C_{2}[i] = k \rightarrow \alpha(C_{1}[i] = k$.
\end{enumerate}
Thus we see that, for a particular $i$, $\forall k \in \{0,
1\} \alpha(C_{2}[i] = k \rightarrow \alpha(C_{1}[i] = k
\wedge \alpha(C_{1}[i] = \mu \rightarrow \alpha(C_{2}[i] =
\mu$. Without loss of generality this holds for all bit
positions $0 < i < n-1$. Hence
\Eqn{eqn_gamma_monotonic_to_prove} holds.
\end{proof}

\proofheading{Property \textnormal{G2}. $\gamma$ is
monotonic, \ie $\forall T_{1}, T_{2} \in \bb{T}{n}, T_{1}
\leqabst T_{2} \rightarrow \gamma(T_{1}) \leqconc
\gamma(C_{2}) $}. 

\begin{proof}
From the definition of the partial order on the abstract
domain (\Eqn{eqn_tnum_ordering_relation} in the main text),
we have for two tnums $T_{1}$ and $T_{2}$.

\begin{align}
\label{eqn_tnum_ordering_relation_appendix}
\begin{split}
T_{1} \leqabst T{2} \; \triangleq & \; \forall i, 0 \leq i 
\leq n-1, \forall k \in \{0, 1\}: \\ 
& (\ithbitx{T_{1}}{i} = \mu \Rightarrow \ithbitx{T_{2}}{i} = \mu)  \\
& \wedge (\ithbitx{T_{1}}{i} = k \Rightarrow \ithbitx{T_{2}}{i} = k \,\veebar\, \ithbitx{T_{2}}{i} = \mu) \\
& \wedge (\ithbitx{T_{2}}{i} = k \Rightarrow \ithbitx{T_{1}}{i} = k) 
\end{split}
\end{align}

We will prove $\gamma(T_{1}) \subseteq \gamma(T_{2})$ by
showing that $\forall x: x \in \gamma(T_{1}) \rightarrow x
\in \gamma(T_{2})$. Consider a particular $x \in
\gamma(T_{1})$. By the defintion of $\gamma$, we have, $x
\;\bitwiseAnd \; \bitwiseNot \;T_{1}.m = T_{1}.v$. Or
alternatively, we have, where $x[i]$ represents the \ith{i}
bit of a value: $\forall i : x[i] \;\bitwiseAnd \;
\bitwiseNot \;T_{1}.m[i] = T_{1}.v[i]$

We consider the following 3 cases: 
\begin{enumerate}

\item $T_{1}.v[i] = 0$, and $T_{1}.m[i] = 1$, \ie ($T_{1}[i]
= \mu$). From \Eqn{eqn_tnum_ordering_relation_appendix} this
implies that $T_{2}.v[i] = 0$, and $T_{2}.m[i] = 1$. Hence
every $x[i]$ that satisfies  $x[i] \;\bitwiseAnd \;
\bitwiseNot \;T_{1}.m[i] = T_{1}.v[i]$ also satisfies $x[i]
\;\bitwiseAnd \; \bitwiseNot \;T_{2}.m[i] = T_{2}.v[i]$
\item $T_{1}.v[i] = 0$, and $T_{1}.m[i] = 0$, \ie ($T_{1}[i]
= 0$). By \Eqn{eqn_tnum_ordering_relation_appendix}, we  
can have the following cases:
\begin{enumerate}
\item $(T_{2}[i] = 0)$ \ie  $T_{2}.v[i] = 0$, and
$T_{2}.m[i] = 0$. Here $T_{2}.v[i] = T_{1}.v[i]$, and
$T_{2}.m[i] = T_{1}.m[i]$.  Hence every $x[i]$ that
satisfies  $x[i] \;\bitwiseAnd \; \bitwiseNot \;T_{1}.m[i] =
T_{1}.v[i]$ also satisfies $x[i] \;\bitwiseAnd \;
\bitwiseNot \;T_{2}.m[i] = T_{2}.v[i]$. 
\item $(T_{2}[i] = \mu)$ \ie  $T_{2}.v[i] = 0$, and
$T_{2}.m[i] = 1$. The $x[i]$ that satisfies $x[i]
\;\bitwiseAnd \; \bitwiseNot \;T_{1}.m[i] = T_{1}.v[i]$ is
only $0$, but the $x[i]$ that satisfies $x[i] \;\bitwiseAnd
\; \bitwiseNot \;T_{2}.m[i] = T_{2}.v[i]$ are $\{0, 1\}$.
Again, every $x[i]$ that satisfies  $x[i] \;\bitwiseAnd \;
\bitwiseNot \;T_{1}.m[i] = T_{1}.v[i]$ also satisfies $x[i]
\;\bitwiseAnd \; \bitwiseNot \;T_{2}.m[i] = T_{2}.v[i]$.
\end{enumerate}
\item $T_{1}.v[i] = 1$, and $T_{1}.m[i] = 1$, \ie ($T_{1}[i]
= 1$). This case is similar to the above case. Again, we
have that every $x[i]$ that satisfies  $x[i] \;\bitwiseAnd
\; \bitwiseNot \;T_{1}.m[i] = T_{1}.v[i]$ also satisfies
$x[i] \;\bitwiseAnd \; \bitwiseNot \;T_{2}.m[i] =
T_{2}.v[i]$.
\end{enumerate}

From the above case analysis, we have that for a particular
bit position $i$, and $x \in \gamma(T{_1})$, $x[i]
\;\bitwiseAnd \; \bitwiseNot \;T_{1}.m[i] = T_{1}.v[i]
\rightarrow x[i] \;\bitwiseAnd \; \bitwiseNot \;T_{2}.m[i] =
T_{2}.v[i]$. Since we are dealing with only bitwise
operations, we can say that this property holds for all bit
positions $i$. Hence, for some $x \in \gamma(T{_1})$, $x
\;\bitwiseAnd \; \bitwiseNot \;T_{1}.m = T_{1}.v \rightarrow
x \;\bitwiseAnd \; \bitwiseNot \;T_{2}.m = T_{2}.v$. Without
loss of generality, this holds for all $x \in
\gamma(T_{1})$. By the definition of $\gamma$, we have
$\forall x: x \in \gamma(T_{1}) \rightarrow x \in
\gamma(T_{2})$. This proves the result. 

\end{proof}

\proofheading{Property \textnormal{G3}. $\gamma \circ
\alpha$ is extensive, \ie $C \in \Zpowerset: C \leqconc
\gamma(\alpha(C))$}. 

\begin{proof}
Consider a particular $C \in \Zpowerset$. This set is
constructed by elements drawn from $\bb{Z}{n}$, \ie $C : \{c
\mid c \in \bb{Z}{n}\}$. Let $T = \alpha(C)$. To show
extensivity, we have to show that: $\forall c \in C: c \in
\gamma(T)$. 

Now, by the definition of $\gamma$ in \Eqn{eqn_gamma}, we
have: $c \in \gamma(T) \Leftrightarrow \forall i, c[k]
\;\bitwiseAnd\; \;\bitwiseNot(T.m)[i] = T.v[i]$.

We consider the following cases:
\begin{enumerate}
\item $\forall c_{a}, c_{b} \in C, k \in \{0, 1\}: c_{a}[i]
= c_{b}[i] = k$. This means that all the elements in the set
$C$ share the same value at bit position $i$. This gives 
use two cases:
\begin{enumerate}
\item $\forall c_{a}, c_{b} \in C: c_{a}[i] = c_{b}[i] = 0$
Consider the tnum $T=\alpha(C)$. By the definition of $\alpha$
(\Eqn{eqn_alpha}), $T.v[i] = 0$, $T.m[i] = 0$. Thus it holds
that  $\forall c \in C, c[i] \;\bitwiseAnd\; \;\bitwiseNot
T.m[i] = T.v[i]$,, since $0 \;\bitwiseAnd\; \;\bitwiseNot 0
= 0$.
\item $\forall c_{a}, c_{b} \in C: c_{a}[i] = c_{b}[i] = 1$
Consider the tnum $T=\alpha(C)$. By the definition of $\alpha$
(\Eqn{eqn_alpha}), $T.v[i] = 1$, $T.m[i] = 0$. Thus it holds
that  $\forall c \in C, c[i] \;\bitwiseAnd\; \;\bitwiseNot
T.m[i] = T.v[i]$, since $1 \;\bitwiseAnd\; \;\bitwiseNot
0 = 1$.
\end{enumerate}
\item  $\exists c_{a}, c_{b} \in C, k \in \{0, 1\}: c_{a}[i]
\neq c_{b}[i] = k$. This means that there are two elements
drawn from the set $C$ that do not share the same value at
bit position $i$. Consider the tnum $T=\alpha(C)$. From the
definition of $\alpha$ (\Eqn{eqn_alpha}), $T.v[i] = 0$, and
$T.m[i] = 1$. Regardless of whether $c[i] = 1$ or $c[i] =
0$, it satisfies holds that $\forall c \in C, c[i]
\;\bitwiseAnd\; \;\bitwiseNot T.m[i] = T.v[i]$, since both
$1 \;\bitwiseAnd\; \;\bitwiseNot 1 = 0$, and $0
\;\bitwiseAnd\; \;\bitwiseNot 1 = 0$.
\end{enumerate}

Since both $\alpha$ and $\gamma$ are bitwise exact, we can
say that for a particular $c \in C$, $\forall i: c[i]
\;\bitwiseAnd\; \;\bitwiseNot T.m[i] = T.v[i]$. This implies 
$ c \;\bitwiseAnd\; \;\bitwiseNot T.m = T.v$. Without loss
of generality, we can say that  $\forall c \in C$, $c
\;\bitwiseAnd\; \;\bitwiseNot T.m = T.v$. This gives us 
the result:  $\forall c \in C: c \in \gamma(\alpha(C))$.
\end{proof}

\proofheading{Property \textnormal{G4}. $\alpha \; \circ\;
\gamma$ is reductive, \ie $T \in \bb{T}{n}:
\alpha(\gamma(T)) \leqabst T$}.

\begin{proof}
Consider a particular $T \in \bb{T}{n}$. Let $C =
\gamma(T)$. If $C : {c_{1}, c_{2}, \ldots c_{n}}$.  We know
from the definition of $\alpha$ that $\alpha(C).v = c_{1}
\;\bitwiseAnd\; c_{2} \;\bitwiseAnd\; \ldots c_{n}$, and
$\alpha(C).m = (c_{1} \;\bitwiseAnd\; c_{2} \;\bitwiseAnd\;
\ldots c_{n}) \;\bitwiseXor \; (c_{1} \;\bitwiseOr\; c_{2}
\;\bitwiseOr\; \ldots c_{n})$. 

To show reductivity, we will show that $\forall T,
\alpha(\gamma(T)) == T$., \ie $\forall T, \alpha(C) == T$,
where $C = \gamma(T)$. In other words we will show that:
$\forall i: \alpha(C).v[i] == T.v[i]$ and that $\forall i:
\alpha(C).m[i] == T.m[i]$. We consider the following 3
possible cases:

\begin{enumerate}
\item $T[i] = \mu$, \ie $T.v[i] ==1, T.m[i] == 0$. Now, $C =
\gamma(T)$. From the definition of $\gamma$, $\forall c \in
C: c[i] = 1$. By the definition of $\alpha$, $\alpha(C).v[i]
= 1$, and $\alpha(C).m[i] = 0$. Hence, we see that
$\alpha(C).v[i] = T.v[i]$ and $\alpha(C).m[i] = T.m[i]$
\item $T[i] = 0$, \ie $T.v[i] ==0, T.m[i] == 0$. Now, $C =
\gamma(T)$. From the definition of $\gamma$, $\forall c \in
C: c[i] = 0$. By the definition of $\alpha$, $\alpha(C).v[i]
= 0$, and $\alpha(C).m[i] = 0$. Hence, we see that
$\alpha(C).v[i] = T.v[i]$ and $\alpha(C).m[i] = T.m[i]$
\item $T[i] = 1$, \ie $T.v[i] ==1, T.m[i] == 0$. Now, $C =
\gamma(T)$. From the definition of $\gamma$, $\forall c \in
C: c[i] = 1$. By the definition of $\alpha$, $\alpha(C).v[i]
= 1$, and $\alpha(C).m[i] = 0$. Hence, we see that
$\alpha(C).v[i] = T.v[i]$ and $\alpha(C).m[i] = T.m[i]$
\end{enumerate}

The above property holds for all bit positions $i$. Hence,
for a given tnum $T$ and $C = \gamma(T)$, we can say that
$\forall i: \alpha(C).v[i] = T.v[i] \;\wedge\;
\alpha(C).m[i] = T.m[i]$. Without loss of generality, 
we can say that $\forall T: \alpha(\gamma(T)) = T$, \ie
$\forall T: \alpha(\gamma(T)) \leqabst T$
\end{proof}

\end{document}